\documentclass[copyright,creativecommons,sharealike,noncommercial]{eptcs}
\providecommand{\event}{WFLP 2016} 
\usepackage{breakurl}             % Not needed if you use pdflatex only.

\usepackage[utf8]{inputenc} 
\usepackage[T1]{fontenc} 

\usepackage{enumitem}

\usepackage{tikz}
\usetikzlibrary{calc,patterns}
\usepackage{tikzscale} % auch für \includegraphics{*.tikz}
\usepackage{pgfplots}
\pgfplotsset{compat=newest}
\usepgfplotslibrary{groupplots}

\pgfplotscreateplotcyclelist{colorlist}{
{color=map1, mark=x, mark size=1.75},
{color=map2, mark=*, mark size=1},
{color=map3, mark=square*, mark size=1},
{color=map4, mark=triangle*, mark size=1},
}

\usepackage{caption}
\usepackage{subcaption} % für subfigures

\usepackage{csquotes} % Anführungszeichen

\usepackage{amsmath} 
\usepackage{amssymb} 
\usepackage{amsthm} 
\usepackage{mathtools} % für \xRightarrow

\newtheorem{theorem}{Theorem}
\newtheorem{lemma}{Lemma}

\theoremstyle{definition}
\newtheorem{definition}{Definition}

\newenvironment{sproof}{
    \proof}{\endproof}

\usepackage{graphicx}
\usepackage{color} 

\definecolor{middlegray}{rgb}{0.3,0.3,0.3}
\definecolor{lightgray}{rgb}{0.7,0.7,0.7}
\definecolor{red}{rgb}{.8,0,0}
\definecolor{blue}{rgb}{0,0,0.5}
\definecolor{lightblue}{rgb}{0.3,0.5,0.7}
\definecolor{yac}{rgb}{0.4,0.4,0.1}
\definecolor{green}{rgb}{0,0.5,0.1}
\definecolor{purple}{rgb}{0.4,0.05,0.25}
\definecolor{yellow}{rgb}{0.8,0.8,0.3}

\definecolor{map1}{RGB}{166,206,227}
\definecolor{map2}{RGB}{31,120,180}
\definecolor{map3}{RGB}{178,223,138}
\definecolor{map4}{RGB}{51,160,44}

\newcommand{\cd}[1]{\lstinline[style=stHaskellInline,keepspaces=true,breaklines=false]$#1$}

\def\lstsmallmath{\leavevmode\ifmmode \scriptscriptstyle \else  \fi}
\def\lstsmallmathend{\leavevmode\ifmmode  \else  \fi}

\usepackage{listings} 

\lstnewenvironment{code}[1][]
   {
     \lstset{style=stHaskell,
      xleftmargin=0pt,
      xrightmargin=0pt,
      #1
    }} {}

\lstdefinestyle{stHaskell}{
      basicstyle=\footnotesize\ttfamily,
      stringstyle=\color{purple}\ttfamily,
      commentstyle=\color{middlegray}\ttfamily\slshape,
      numberstyle={\ttfamily},
      lineskip=-0.1ex,
      belowskip=8pt,
      frame=single,
      escapebegin={\lstsmallmath}, escapeend={\lstsmallmathend},
      language=Haskell,
      keywordstyle=[1]\color{yac}\ttfamily,
      deletekeywords=[1]{putStr,map,array,data,random,lenght,zipWith,concat,uncurry,zip,reverse,length,div,List,System,Random,getArgs,read,putStrLn,show,take,randoms,mkStdGen},
      morekeywords=[2]{prodBList},
      keywordstyle=[2]\color{red}\ttfamily,
      morekeywords=[3]{bMerge},
      keywordstyle=[3]\color{blue}\ttfamily,
      otherkeywords={},
      morekeywords=[4]{bSplit},
      keywordstyle=[4]\color{green}\ttfamily,
      emph=[1]{bSort_0}, 
      emphstyle=[1]\color{blue},
      emph=[2]{bitonic_merger_0}, 
      emphstyle=[2]\color{red},
      literate={+}{{$\scriptstyle +$}}1 {/}{{$/$}}1 {*}{{$*$}}1 {=}{{$=$}}1
               {>}{{$>$}}1 {<}{{$<$}}1 {\\}{{$\lambda$}}1 {\\n}{{\textbackslash
               n}}2 {\\\\}{{\char`\\\char`\\}}1
               {->}{{$\rightarrow$}}2 {>=}{{$\geq$}}2 {<-}{{$\leftarrow$}}2
               {<=}{{$\leq$}}2 {=>}{{$\Rightarrow$}}2
               {\ .}{{$\scriptstyle \circ$}}2 {\ .\ }{{\  $\scriptstyle\circ$\ }}2
               {>>}{{>>}}2 {>>=}{{>>=}}2
               {|}{{$\mid$}}1
               {Ö}{{\"O}}1
               {Ä}{{\"A}}1
               {Ü}{{\"U}}1
               {ß}{{\ss}}2
               {ü}{{\"u}}1
               {ä}{{\"a}}1
               {ö}{{\"o}}1
           } 

\lstdefinestyle{stHaskellInline}{
      basicstyle=\ttfamily,
      keywordstyle=\ttfamily,
      lineskip=-0.1ex,
      frame=single,
      language=Haskell,
      literate={+}{{$+$}}1 {/}{{$/$}}1 {*}{{$*$}}1 {=}{{$=$}}1
               {>}{{$>$}}1 {<}{{$<$}}1 {\\}{{$\lambda$}}1 {\\n}{{\textbackslash
               n}}2 {\\\\}{{\char`\\\char`\\}}1
               {->}{{$\rightarrow$}}2 {>=}{{$\geq$}}2 {<-}{{$\leftarrow$}}2
               {<=}{{$\leq$}}2 {=>}{{$\Rightarrow$}}2
               {\ .}{{$\circ$}}2 {\ .\ }{{\ $\circ$\ }}2
               {>>}{{>>}}2 {>>=}{{>>=}}2
               {|}{{$\mid$}}1
               {|}{{$\mid$}}1
               {Ö}{{\"O}}1
               {Ä}{{\"A}}1
               {Ü}{{\"U}}1
               {ß}{{\ss}}2
               {ü}{{\"u}}1
               {ä}{{\"a}}1
               {ö}{{\"o}}1
}

\DeclareMathOperator{\s}{sN}
\DeclareMathOperator{\concat}{concat}

\title{An Agglomeration Law for Sorting Networks and its Application in Functional Programming}  

\author{Lukas Immanuel Schiller
\institute{Philipps-Universit\"at Marburg}
\institute{Fachbereich Mathematik und Informatik}
\email{schiller@mathematik.uni-marburg.de}}

\def\titlerunning{Agglomeration Law for Sorting Networks}
\def\authorrunning{Lukas Schiller}

\begin{document}
\maketitle

\begin{abstract}
    In this paper we will present a general agglomeration law for sorting networks. Agglomeration is a common technique when designing parallel programmes to control the granularity of the computation thereby finding a better fit between the algorithm and the machine on which the algorithm runs. Usually this is done by grouping smaller tasks and computing them en bloc within one parallel process. In the case of sorting networks this could be done by computing bigger parts of the network with one process. The agglomeration law in this paper pursues a different strategy: The input data is grouped and the algorithm is generalised to work on the agglomerated input while the original structure of the algorithm remains. 
    This will result in a new access opportunity to sorting networks well-suited for efficient parallelization on modern multicore computers, computer networks or GPGPU programming. Additionally this enables us to use sorting networks as (parallel or distributed) merging stages for arbitrary sorting algorithms, thereby creating new hybrid sorting algorithms with ease. The expressiveness of functional programming languages helps us to apply this law to systematically constructed sorting networks, leading to efficient and easily adaptable sorting algorithms. An application example is given, using the Eden programming language to show the effectiveness of the law. The implementation is compared with different parallel sorting algorithms by runtime behaviour. % with runtime evaluations.   %and substantiated by runtime comparisons 
\end{abstract}

\section{Introduction}
With the increased presence of parallel hardware the demand for parallel algorithms increases accordingly. Naturally this demand includes sorting algorithms as one of the most interesting tasks of computer science. A particularly interesting class of sorting algorithms for parallelization is the class of {\it oblivious} algorithms. We will call a parallel algorithm oblivious ``iff its communication structure and its communication scheme are the same for all inputs the same size''~\cite{lowe}.

Sorting networks are the most important representative of the class of oblivious algorithms. They have been an interesting field of research since their introduction by Batcher~\cite{batcher68} in 1968 and are experiencing a renaissance in GPGPU programming~\cite{gpu}.  
They are based on comparison elements, mapping their inputs $(a_1, a_2) \mapsto (a_1', a_2')$ with  $a_1' = \min(a_1,a_2)$ and $a_2' = \max(a_1,a_2) $ and therefore $a_1' \leq a_2'$. A simple graphical representation is shown in \autoref{fig:compelem_00}. The arrowhead in the box indicates where the minimum is output. 

\begin{figure}[htb]
    \centering
    \includegraphics[height=0.1\textwidth]{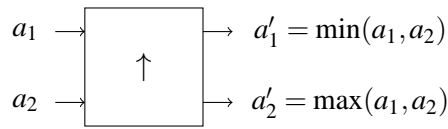}
    \caption{Comparison element (ascending).}
    \label{fig:compelem_00}
\end{figure}

A simple functional description of sorting networks results in a repeated application of this comparison element function with fixed indices for every step. For a sequence $(a_1, \ldots, a_n)$ of length $n$ the specific steps are fixed: $$ (a_1, \ldots, a_n) \mapsto \ldots \mapsto (\tilde{a_1},\ldots,a_i,\ldots,a_j,\ldots,\tilde{a_n}) \mapsto (\tilde{a_1},\ldots,a_i',\ldots,a_j',\ldots,\tilde{a_n}) \mapsto \ldots \mapsto (a_1',\ldots,a_n') $$
with $i \neq j $. In a specific step $a_i$ and $a_j$ are sorted with a comparison element. Ultimatly resulting in the sorted sequence $(a_1',\ldots, a_n')$.

\autoref{fig:compsimplesort} shows a simple sorting network for lists of length 4. For every permutation of the input $ (a_1, \dots, a_4) $  the output $ (a'_1, \dots, a_4') $ is sorted -- the comparisons are independent of the data base. 
Notice the obvious inherent parallelism in the first two steps of the sorting network. The
restriction to a fixed structure of comparisons results in an easily predictable behaviour and easily detectable parallelism.

\begin{figure}[htb]
    \centering
    \includegraphics[width=0.8\textwidth]{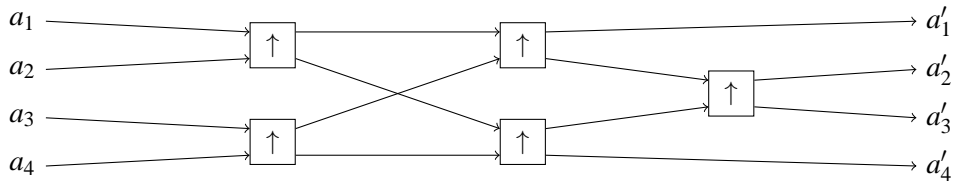}
    \caption{Simple sorting network with comparison elements. Source: \protect \cite{Knuth98}}
    \label{fig:compsimplesort}
\end{figure}

Some well-known sorting algorithms, for example Bubble Sort~\cite{Knuth98}, can be described as sorting networks. Especially in the case of recursively constructed sorting networks (e.g.\ Batcher's Bitonic Sort or Batcher's Odd-Even-Mergesort), with their inherent functional structure, an obviously correct description of the algorithm is easily possible in a functional programming language such as Haskell~\cite{haskell2010}.  

In practice straightforward implementations of these algorithms often struggle with too fine a granularity of computation and therefore do not scale well. Agglomerating parts of the algorithm is a common step in dealing with this problem when designing parallel programmes (compare Foster's PCAM method~\cite{foster}). 
With recursive algorithms for example it is a common technique to agglomerate branches of the recursive tree by parallelising only until a specific depth of recursion. 
With a coarser granularity the computation to communication ratio improves. A common agglomeration for sorting networks is to place blocks or rows of comparison elements in one parallel process. 

In this paper we discuss a different approach. We will agglomerate the input data and alter the comparison element to work on blocks of data. This approach is not based upon the structure of a specific sorting network and can therefore be applied to any sorting network. At the same time we will see that the limited nature of sorting networks is necessary for this modification to be correct. The application of this transformation will open a different access to sorting networks, allowing easy combination with other sorting algorithms. Working on data structures instead of single elements leads to a suitable implementation for modern multi-core computers, GPGPU concepts or computer networks. We will obtain an adequate granularity of computation and the width of the sorting network can correspond with the number of processor units. A second layer of traditional agglomeration (e.g.\ blocks or rows of comparison elements) is independently possible. 

In \autoref{agglo} we will discuss which demands are necessary for altered comparison elements to preserve an algorithm's functionality and correctness. In \autoref{bsp} an example is given showing situations in which the application of this agglomeration is beneficial and tests with different approaches are evaluated. \autoref{relwork} discusses related work 
and \autoref{conc} concludes.

\section{Agglomeration Law for Sorting Networks}
\label{agglo}
In general, sorting networks work on sequences of elements $A = (a_1,\ldots,a_n)$.  
Our improvement will work with a partition of a given sequence. 
In the following, we will use Haskell notation and lists instead of sequences to improve readability, even though a more general type would be possible.

\begin{theorem}[Agglomeration Law for Sorting Networks]
    \label{theo:1}
    Let $\mathtt{A = [a_1,\ldots,a_n]}$ be a sequence where a total order ``$\leq$'' is defined on the elements,  
    $\mathtt{c :: Ord\; a \Rightarrow (a,a) \rightarrow (a,a)}$ a comparison element as described before and 
    $$\text{\cd{sN :: ((b,b) -> (b,b)) -> [b] -> [b]}}$$
    a correct sorting network, meaning $\mathtt{sN\; c\; A = A'}$ with $\mathtt{A' = [a_1', \ldots , a_n']}$ and $\mathtt{a_1' \leq \ldots \leq a_n'}$ where $\mathtt{a_1',\ldots,a_n'}$ is a permutation of $\mathtt{a_1,\ldots,a_n}$ and the only operation used by the sorting network is a repeated application of the comparison element with a fixed, data independent structure for a given input size.
    Then there exists a comparison element $\mathtt{\it c' :: Ord\; a \Rightarrow ([a],[a]) \rightarrow ([a],[a])}$ with which a sequence of sequences $\mathfrak{A} = [A_1,\ldots,A_n]$ with $A_i = [a_{i1},\ldots,a_{in_{i}}]$ can be sorted with the same sorting network. 
    %the for every partition of a given sequence  
    Meaning that $\mathtt{\s\; c'\; \mathfrak{A} = \mathfrak{A'}}$ with $\mathtt{\mathfrak{A'} = [A_1',\ldots,A_n']}$ and $\mathtt{A_1' \preceq \ldots \preceq A_n'}$. Where $\concat\; \mathfrak{A'}$ is a permutation of $\concat\; \mathfrak{A}$ and $A \preceq B$ means that for two sequences $A = [a_1,\ldots,a_p]$ and $B = [b_1,\ldots,b_q]$ every element of $A$ is less than or equal to every element from $B$: $$ A \preceq B \Leftrightarrow \forall a \in A, \forall b \in B : a \leq b $$ 
\end{theorem}
With blocks of data the concatenation of the elements of $\mathfrak{A'}$ need to be a permutation of the concatenation of $\mathfrak{A}$, the elements themselves ($\mathtt{A_1',\ldots,A_n'}$) do not need to be a permutation of $\mathtt{A_1,\ldots,A_n}$. 

Note that the order relation for blocks of data ``$\preceq$'' defines only a partial order whereas the elements inside the blocks are totally ordered. 
To this end we need to specialise  
the comparison element to deal with the case of overlapping or encasing blocks and still fulfill all properties necessary for the sorting network to work correctly (cf.\ \autoref{fig:overlap}). 

\begin{figure}[htb]
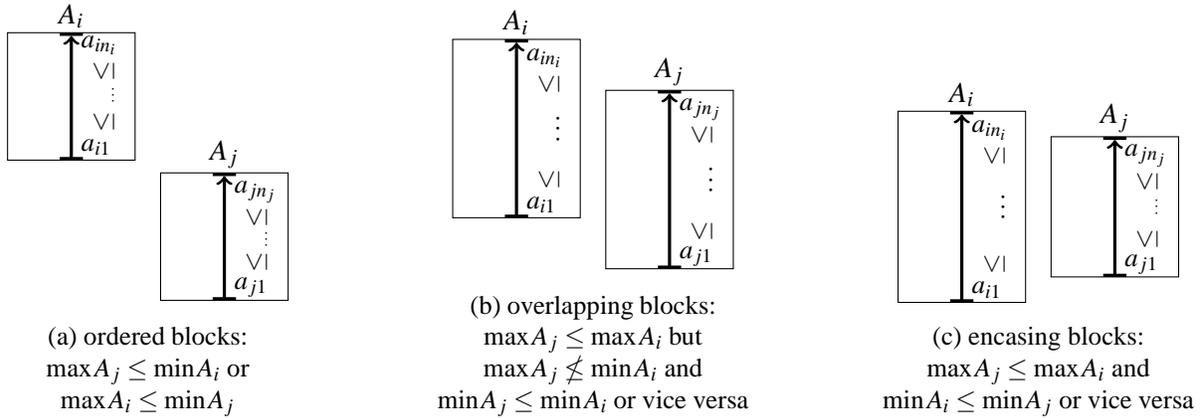

    \centering
    \begin{subfigure}[b]{0.26\textwidth}
        \centering
        \captionsetup{justification=centering}
        \includegraphics[width=0.9\linewidth]{figures/zerlegung_bloecke3.tikz}
        \caption{ordered blocks:\\ $\max A_j \leq \min A_i$ or  $\max A_i \leq \min A_j$}
        \label{fig:overlap0}
    \end{subfigure}
    \hfill
    \begin{subfigure}[b]{0.26\textwidth}
        \centering
        \captionsetup{justification=centering}
        \includegraphics[width=0.9\linewidth]{figures/zerlegung_bloecke0.tikz}
        \caption{overlapping blocks:\\ $\max A_j \leq \max A_i$ but $\max A_j \nleq \min A_i$ and $\min A_j \leq \min A_i$ or vice versa}
        \label{fig:overlap1}
    \end{subfigure}
    \hfill
    \begin{subfigure}[b]{0.26\textwidth}
        \centering
        \captionsetup{justification=centering}
        \includegraphics[width=0.9\linewidth]{figures/zerlegung_bloecke1.tikz}
        \caption{encasing blocks:\\ $\max A_j \leq \max A_i$ and $\min A_i \leq \min A_j$ or vice versa}
        \label{fig:overlap2}
    \end{subfigure}
    \caption[Overlapping and encasing blocks]{Cases for comparison elements for blocks of data: blocks can be ordered (with order relation $\preceq$), overlapping or encasing, where overlapping and encasing means that the blocks are not in an order relation between one another (meaning neither $\preceq$ nor $\succeq$ holds). }
    \label{fig:overlap}
\end{figure}

If for example the input lists overlap (e.g.\ $\mathtt{c'\; ([1,2,3,4],[3,4,5,6])}$) a simple swap would not fulfill the requirements. We would prefer a result such as $\mathtt{([1,2,3,3],[4,4,5,6])}$ and therefore $\mathfrak{A'}$ can not be a permutation of $\mathfrak{A}$ but we expect that every element $\mathtt{a_{ij}}$ from $\mathtt{A_1,\ldots,A_n}$ is in $\mathtt{A_1',\ldots,A_n'}$.
In the next step we will investigate which conditions a comparison element for blocks of data must fulfill. 

\setcounter{equation}{-1}
\subsection{Comparison element for partially ordered blocks of totally ordered elements}
If we want to alter the comparison element while preserving the functionality and correctness of the sorting network we must understand which information is generated and preserved within a traditional comparison element. We will therefore investigate the capabilities and limits of comparison elements for totally ordered sequences: Let $a_1, a_2, a_1^1, a_1^2, a_2^1, a_2^2$ be elements where information about the following relations have already been gathered by the sorting network: \begin{equation} a_1^1 \leq a_1 \leq a_1^2 \text{\hspace{2em} and\hspace{2em} } a_2^1 \leq a_2 \leq a_2^2 \label{eq:0} \end{equation}
If we do sort $a_1$ and $a_2$ with a comparison element $(a_1,a_2) \mapsto (a_1', a_2')$ we receive new relations (e.g.\ $a_1^1 \leq a_1 \Rightarrow a_1^1 \leq a_2'$). We will distinguish between {\it direct relations} and {\it conditional relations}. In this case direct relations refer to all relations resulting directly from the relations which exists and are known before the application of the comparison element %direct resulting relations which are valid %fixme fixme fixme 
%in any case 
and which involve $a_1,a_2,a_1' \text{ or } a_2'$. We expect the comparison element to be side-effect free and therefore we expect every relation between elements not touched by the comparison element to be unaffected by its application. Here the direct relations resulting from (\ref{eq:0}) are: 
\begin{eqnarray}  
    a_1' &\leq& a_2' \label{eq:1} \\
    a_1' &\leq& a_i^2, \; i \in \{1,2\} \label{eq:2}\\
    a_i^1 &\leq& a_2', \; i \in \{1,2\} \label{eq:3}
\end{eqnarray}
If we have additional information, we get additional direct relations. For $\{i,j\} = \{1,2\}$:
\begin{eqnarray}  
a_i^1 \leq a_j &\Rightarrow& a_i^1 \leq a_1' \label{eq:4} \\
a_j \leq a_i^2 &\Rightarrow& a_2' \leq a_i^2 \label{eq:5} \\
a_i \leq a_j   &\Rightarrow& a_i^1 \leq a_1' \land a_2' \leq a_j^2 \label{eq:direct} 
\end{eqnarray}
We call these relations \emph{direct relations} only if the left side is already known. 

\begin{definition}[Valid comparison elements for blocks of data]
   %Let $\mathtt{A_1 = [a_{11},\ldots,a_{1n}], A_2 = [a_{21},\ldots,a_{2m}]}$ 
    Let $\mathtt{A_1, A_2} $ % 
    be sequences with a total order ``$\leq$'' defined on the elements and $\mathtt{c' :: Ord\; a \Rightarrow ([a],[a]) \rightarrow ([a],[a])}$ a block comparison element with $\mathtt{c' (A_1, A_2) = (A_1', A_2')}$ and 
    $A_1' \preceq A_2'$.
    
    $\mathtt{c'}$ is called \emph{valid}, iff all elements from $A_1$ and $A_2$ which are less than $\mathtt{ lb = \max(\min(A_1),\; \min(A_{2}))}$ 
   must be in $A_1'$, all elements which are greater than $\mathtt{ub = \min(\max(A_1),\; \max(A_{2}))}$  must be in $A_2'$ and all elements between these limits can be either in $A_1'$ or in $A_2'$ as long as every element in $A_1'$ is smaller than or equal to every element in $A_2'$ (cf.\ \autoref{fig:mergeSplit00}). 
\end{definition}
\begin{figure}[ht]
    \captionsetup{singlelinecheck=off}
    \centering
    \includegraphics[width=0.30\textwidth]{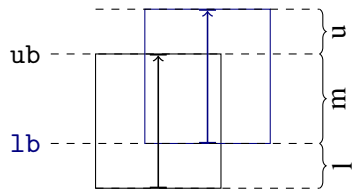}
    \caption{Sections of the comparison element for blocks of data. Elements from $l$ must be in the lesser result ($A_1'$), elements from $u$ must be in the greater result ($A_2'$) and elements from $m$ can be in either result as long as $A_1' \preceq A_2'$ holds. }
    \label{fig:mergeSplit00}
\end{figure}

\begin{lemma}
    \label{lemma:comp}
   \emph{Valid} block comparison elements maintain the \emph{direct relations} fulfilled by the elementary comparison elements. 
\end{lemma}

\begin{proof}
    We show the validity of the above relations (\ref{eq:1}) to (\ref{eq:direct}) for blocks of data:
%    \leavevmode
%    \mbox{}
    \begin{enumerate}[labelindent=35pt,labelwidth=0.75em,leftmargin=!]
        \item $\mathtt{A_1' \preceq A_2'}$ is included in the definition.
        \item $\mathtt{\max A_1' \leq ub \leq \max A_i \leq \min A_i^2 \Rightarrow A_1' \preceq A_i^2, \; i \in \{1,2\} }$
        \item $\mathtt{\max A_i^1 \leq \min A_i \leq lb \leq \min A_2' \Rightarrow  A_i^1 \preceq A_2', \; i \in \{1,2\} }$
        \item $\mathtt{A_i^1 \preceq A_1 \land A_i^1 \preceq A_2 \Rightarrow A_i^1 \preceq [\min (\min A_1,\; \min A_2)] \preceq A_1' \Rightarrow A_i^1 \preceq A_1' , \; i \in \{1,2\}   }$
        \item $\mathtt{A_1 \preceq A_i^2 \land A_2 \preceq A_i^2 \Rightarrow A_2' \preceq [\max (\max A_1,\; \max A_2)] \preceq A_i^2 \Rightarrow A_2' \preceq A_i^2 , \;i \in \{1,2\} }$
        \item $\mathtt{A_i^1 \preceq A_i \preceq A_j \Rightarrow A_i^1 \preceq A_i \land A_i^1 \preceq A_j \Rightarrow \max A_i^1 \leq \min(\min A_i,\; \min A_j) \Rightarrow A_i^1 \preceq A_1'}$ \\
              $\mathtt{A_i \preceq A_j \preceq A_j^2 \Rightarrow A_i \preceq A_j^2 \land A_j \preceq A_j^2 \Rightarrow  \max(\max A_i,\; \max A_j) \leq \min A_j^2 \Rightarrow A_j' \preceq A_j^2}$

    \end{enumerate}
\end{proof}

The proof shows that these limits are not only sufficient but necessary to guarantee the \emph{direct relations} on which sorting networks are essentially based. Counterexamples where a different limit selection leads to the failure of the sorting network can easily be found. 

All other producible information concerns {\it conditional relations} which depend on a yet unknown condition resulting in a disjunction or a conditional with unknown antecedent. 
For example $$(a_1 \leq a_2 \lor a_2 \leq a_1) \land a_1^1 \leq a_1 \leq a_1^2 \Rightarrow a_1^1 \leq a_1' \lor a_2' \leq a_1^2$$
For ordered or overlapping blocks we can easily verify that all these relations can be preserved, as every input element has a direct descendant, analogous to the original comparison element. In this case a {\it direct descendant} $\mathtt{A'}$ of a block $\mathtt{A}$ is bounded by the extrema of the parental block, meaning that $\mathtt{\min A \leq \min A'}$ and $\mathtt{\max A' \leq \max A}$. $\mathtt{A'}$ can but need not contain elements from $\mathtt{A}$ as well as elements which are not in $\mathtt{A}$ due to the fact that the property is defined through boundaries not elements. Therefore, when applying the comparison element, the boundaries of each block can at the most approach each other, leaving all relations preserved. An example is given in \autoref{fig:zerlegung3}.  
\begin{figure}[htb]
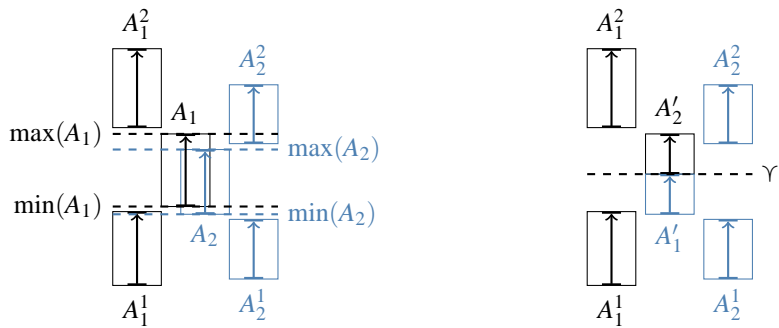

    \centering
    \begin{subfigure}[b]{0.348\textwidth}
        \centering
        \includegraphics[height=0.8\linewidth]{figures/zerlegung_block_05.tikz}
    \end{subfigure}
    \hspace{8mm}
    \begin{subfigure}[b]{0.348\textwidth}
        \centering
        \includegraphics[height=0.8\linewidth]{figures/zerlegung_block_06_neu.tikz} 
    \end{subfigure}
    \caption{Split of overlapping blocks. In this case the minimal (maximal) element of $A_2$ is smaller than the minimal (maximal) element of $A_1$. Thereby $A_2$ ``shrinks'' from above, meaning that the maximum element of $A_1'$ is smaller than $\max A_2$. This does not yet disclose any information about the number of elements in $A_1'$. $A_1$ ``shrinks'' from below. We can see $A_1'$ as the descendant of $A_2$ and $A_2'$ as the descendant of $A_1$. All relations are preserved.  } 
    \label{fig:zerlegung3}
\end{figure}

%With the idea of a direct descendent it is easy to understand that every relevant relation is preserved as the situation is analogous to the original case. 
With encased blocks (cf.\ \autoref{fig:overlap2}) it is not necessarily possible to find a descendant for every element. If, for example, we have 
$ A_1^1 \preceq A_1 \preceq A_1^2 $ and $ A_2^1 \preceq A_2 \preceq A_2^2$ there might be no output element $A_i'$ with $A_1^1 \preceq A_i' \preceq A_1^2$ (cf.\ \autoref{fig:zerlegung3-2}). 
\begin{figure}[htb]
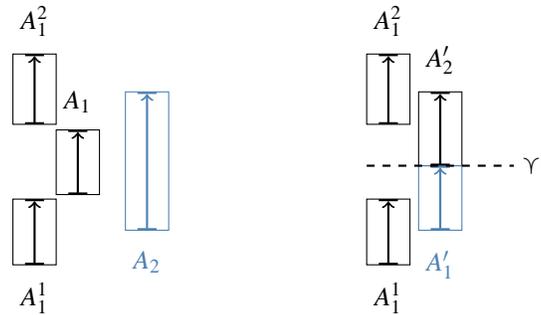

    \centering
    \begin{subfigure}[b]{0.3\textwidth}
        \centering
        \includegraphics[height=0.9\linewidth]{figures/pivotblock1.tikz}
    \end{subfigure}
    \begin{subfigure}[b]{0.3\textwidth}
        \centering
        \includegraphics[height=0.9\linewidth]{figures/pivotblock2.tikz} 
    \end{subfigure}
    \caption{Split of encased blocks. There are no direct descendants. $A_1' \preceq A_1^2$ and $A_1^1 \preceq A_2'$ but neither $A_1'$ nor $A_2'$ is between $A_1^1$ and $A_1^2$.} 
    \label{fig:zerlegung3-2}
\end{figure}

Unfortunately, a consequence of this is that this technique of merging and splitting blocks can not necessarily be transferred to a more general sorting algorithm. 
In particular this does not work with pivot based sorting algorithms. However it does with sorting networks because the comparison element does not compare one fixed element with another element but rather returns two sorted elements for which we do not know which input element is mapped to which output element. The information $A_1^1 \prec A_1$ is reduced to $A_1^1 \preceq A_2'$ plus some {\it conditional} information.
Some of this conditional information can no longer be guaranteed to hold but can not be used in a sorting network at all because of the limited operations of sorting networks. The relations of concern are  
\begin{equation}
    a_1 \leq a_2 \lor a_2 \leq a_1 \Rightarrow a_i^1 \leq a_1' \lor a_2' \leq a_i^2,\; i \in \{1,2\}
    \label{eq:6}
\end{equation}
%resulting from $a_i \leq a_j \Rightarrow (a_i^1 \leq a_1' \land a_2' \leq a_j^2), \; \{i,j\} = \{1,2\}$ and $a_1 \leq a_2 \lor a_2 \leq a_1$.

Sorting networks as described above can not produce the additional information needed for this conditional information to become useful.

\begin{lemma}
    \label{lemma:comp2}
    Information about the \emph{conditional relations} (\ref{eq:6}) that can not be preserved by the altered comparison element $\mathtt{c'}$ can not be used by a sorting network. 
\end{lemma}

\begin{sproof}
    The condition of a \emph{conditional relation} is unknown by definition -- otherwise it would be a direct relation. Therefore the implication can not be used to gather additional information. The remaining disjunction can result in useful information in a non-trivial way only if one side of the disjunction is known to be false (modus tollendo ponens) or if both sides of the disjunction are equal. It is not possible to equalise an output element of the comparison element with another element and therefore it is not possible to test whether $a_i \leq a_j$ or not. In particular the information that $a_i \nleq a_j$ can not be produced for any $i$ and $j$. We can not test whether one side of such an equation is false or if both sides are equal and therefore can not use \emph{conditional relations}. Non-trivial, productive information from these disjunctions can only be used in non-oblivious algorithms.
%    All information generated by the sorting network are direct relations ($a_1 \leq a_j$) or of the form \begin{enumerate}
%        \item $a_i \leq a_j \Rightarrow a_k \leq a_l$ (conditional relations I)
%        \item $a_i \leq a_j \lor a_k \leq a_l$ (conditional relations II). 
%    \end{enumerate} Therefore \emph{conditional relations I} can not be used.  
%    The conditional relation (\ref{eq:6}) does only exist if we do not know whether or not $a_1 \leq a_2$ or vice versa. Otherwise it is the case of (\ref{eq:direct}), a direct relation. 
\end{sproof}

\autoref{lemma:comp} and \autoref{lemma:comp2} imply that a comparison element $\mathtt{c'}$ as demanded in \autoref{theo:1} exists with the given limitations from \autoref{lemma:comp}. Therefore all usable information is preserved and this technique of merging and splitting two blocks in a comparison element can be used with every sorting network.  

If the elements inside the blocks are sorted, we can define a linear time comparison element that splits the two blocks into blocks as equal in size as possible. An implementation of such a comparison element can be found in \autoref{bsp}, \autoref{code:sce}. 
Balancing the blocks is advantageous in many cases because it limits the maximal block size to the size of the largest block in the initial sequence. This is beneficial especially in the situation of limited memory for different parts of the parallelised algorithm, for example if the parallelization is done with a computer cluster. By preserving the inner sorting of the blocks, the resulting sequence of the sorting network can be easily combined to a completely sorted sequence by concatenation.   

Every suitable sorting algorithm can be used for the initial sorting inside the blocks. Consequently the sorting network can be used as a skeleton to parallelise arbitrary sorting algorithms and work as the merging stage of the newly combined (parallel) algorithm. A concept that will prove it's worth in the following example. 

\section{Application of the Agglomeration Law on the Bitonic Sorter}   
\label{bsp}
We will now apply the agglomeration law to Batcher's Bitonic Sorting Network. %\footnote{The complete source code of this example can be found at http://www.mathematik.uni-marburg.de/~schiller/agglomeration_law % FIXME 
    It is a recursively constructed sorting network that works in two steps. In the first step an unsorted sequence (of length $2^l$ with $l \in \mathbb{N}$) is transformed into a bitonic sequence. A bitonic sequence is the juxtaposition of an ascending and a descending monotonic sequence or the cyclic rotation of the first case (\autoref{fig:bitoniclists}). 

    \begin{figure}[htb]
        \centering
        \begin{tikzpicture}
            %\footnotesize
            \begin{axis}[axis lines=middle, ticks=none, xmin=0, ymin=0, xmax=10, ymax=10, axis line style={->}, width=0.25\textwidth]
                \addplot[
                    black,
                    every mark/.append style={solid, fill=lightgray}, mark=square*, mark size = 1,
                    nodes near coords, 
                ] coordinates {(1,1) (2,4) (3,5) (4,8) (5,7) (6,6) (7,3) (8,2)};
            \end{axis}
        \end{tikzpicture}
        \hspace{8mm}
        \begin{tikzpicture}
            %\footnotesize
            \begin{axis}[axis lines=middle, ticks=none, xmin=0, ymin=0, xmax=10, ymax=10, axis line style={->}, width=0.25\textwidth]
                \addplot[
                    black,
                    every mark/.append style={solid, fill=lightgray}, mark=square*, mark size = 1,
                    nodes near coords, 
                ] coordinates {(1,3) (2,2) (3,1) (4,4) (5,5) (6,8) (7,7) (8,6)};
            \end{axis}
        \end{tikzpicture}
        \caption{Examples of bitonic sequences.}
        \label{fig:bitoniclists}
    \end{figure}
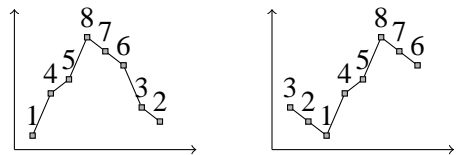

The bitonic sequence is thereafter sorted by a {\em Bitonic Merger}. We will call the function implementing this Bitonic Merger \cd{bMerge} and the function transforming an unsorted sequence into a bitonic sequence \cd{prodBList}. The Bitonic Sorter works with the nested divide-and-conquer scheme of the sorting-by-merging idea. This means that the repeated generation of shorter sorted lists is done by Bitonic Sorters of smaller size.  
A Bitonic Sorter for eight input elements is depicted in \autoref{fig:bitmisch8}. 
\begin{figure}[htb]
    \centering
    \includegraphics[width=0.90\textwidth]{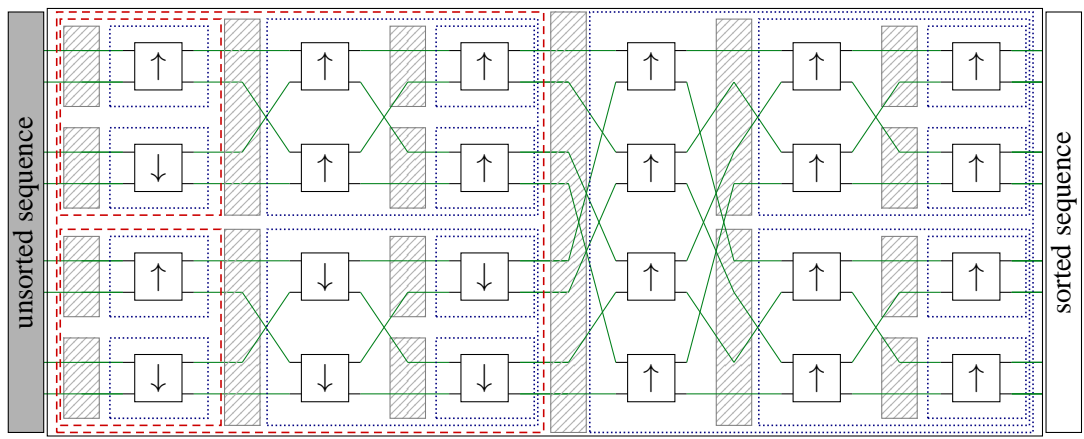}
    \caption{Bitonic Sorter of order 8. The function \cd{prodBList} is represented by a red dashed rectangle, the function \cd{bMerge} by a blue dotted one. Bitonic Sequences are represented by shaded rectangles.}
    \label{fig:bitmisch8}
\end{figure}

The basic component of the sorting network -- the original comparison element -- can be defined as:
\begin{lstlisting}[language=Haskell,breaklines=true,xleftmargin=0pt,
xrightmargin=0pt, style=stHaskell, linerange=Komparatorelement-\-\-\ Hilfsfunktionen, includerangemarker=false,caption={Original comparison element}]
import System.Environment(getArgs)
import Control.Parallel.Eden.Auxiliary
import System.Random
import Data.Time.Clock (diffUTCTime, getCurrentTime)

import Data.List
import Data.Tuple

import Control.Parallel.Eden --rnf

main = do
    args <- getArgs
    let size = read (args!!0) :: Int
    let seed = 42
    let rs = randomlist size seed
    rnf rs `pseq` putStr "" -- "Liste generiert"
    time1 <- getCurrentTime
    let x = bitonicSort rs
    rnf x `pseq` putStr "" -- "Liste sortiert"
    time2 <- getCurrentTime
    --putStrLn $ (show size)++","++
    putStrLn $ (show (diffUTCTime time2 time1)) 
 
-- generating list of random numbers 
randomlist :: Int -> Int -> [Int]
randomlist n seed = take n $ randoms $ mkStdGen seed


-- Aufruf 
bitonicSort :: (Ord a) => [a] -> [a]
bitonicSort = bSort Up 

-- Sortierer
bSort :: Ord a => Direction -> [a] -> [a] -- for lists of length 2^k
bSort _ []  = []   
bSort _ [x] = [x]  
bSort d xss = (bMerge d) . prodBList $ xss where 
    prodBList :: Ord a => [a] -> [a]
    prodBList = unSplit . map (uncurry bSort) . zip [Up, Down] . splitHalf 

-- Mischer
bMerge :: (Ord a) => Direction -> [a] -> [a] -- for lists of length 2^k 
bMerge d [x,y] = compElem d (x,y)
bMerge d xss = unSplit . map (bMerge d) . bSplit $ xss
    where bSplit :: Ord a => [a] -> [[a]]
          bSplit = splitHalf . shuffle . map (compElem d) . perfectShuffle

-- Komparatorelement
data Direction = Up | Down deriving Show

compElem :: Ord a => Direction -> [a] -> [a]
compElem Up [x,y] = if x <= y then [x,y] else [y,x]
compElem Down xs = reverse $ compElem Up xs

-- Hilfsfunktionen
splitHalf :: [a] -> [[a]]
splitHalf = splitIntoN 2

-- Stones perfect shuffle
perfectShuffle :: [a] -> [[a]] 
perfectShuffle xs = unshuffle halfSize xs
                        where halfSize = (length xs) `div` 2
\end{lstlisting}

We will use a two-element-list variant instead of pairs for reasons of code elegance. We define the actual algorithm using the Eden~\cite{loogen05,loogen11} programming language which extends Haskell by the concept of parallel {\it processes} with an implicit communication as well as a Remote Data~\cite{rd10} concept. We can instantiate a process that is defined by a given function with \cd{(\$\#)}:

\begin{code}
($#) :: (Trans b, Trans a) => 
        (a -> b)     -- Process function
        -> a -> b    -- Process input and output 
\end{code}

The class \cd{Trans} consists of {\it transmissible} values. The expression \cd{f \$\# expr} with some function \cd{f :: a -> b} will create a (remote) child process. The expression \cd{expr} will be evaluated (concurrently by a new thread) in the parent process and the result \cd{val} will be sent to the child process. The child process will evaluate \cd{f \$ val} (cf.\ \autoref{fig:process}).

\begin{figure}[htb]
    \centering
    \includegraphics[width=0.56\textwidth]{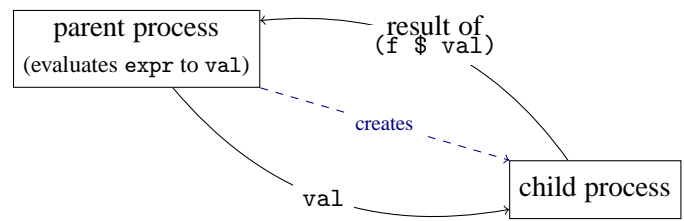} % .46
    \caption[Schema Prozessinstanziierung]{The scheme for process instantiation. Source: \protect \cite{loogen11}}
    \label{fig:process}
\end{figure}

Hereafter we will essentially use Eden's \cd{parMapAt}, a parallel variant of \cd{map} with explicit placement of processes on processor elements (PEs), also called (logical) machines, which are numbered from 1 to the number of processor elements. 

\begin{code}
parMapAt :: (Trans a, Trans b) => 
               [Int]      -- ^places for instantiation          
            -> (a -> b)   -- ^worker function
            -> [a] -> [b] -- ^task list and ^result list
\end{code}

The explicit placement is determined by the first argument, a list of PE numbers specifying the places where the processes will be deployed.
Additionally we will use the constants \cd{noPe} and \cd{selfPe} provided by Eden to calculate the correct placements:  

\begin{code}
noPe   :: Int -- Number of (logical) machines in the system
selfPe :: Int -- Local machine number (ranges from 1 to noPe)
\end{code}

For our implementation we will place each comparison element of the same row on the same PE. In \autoref{code:bit2-bsort} a parallel definition of the algorithm is given. 
\begin{lstlisting}[language=Haskell,breaklines=true,xleftmargin=6pt,breakindent=32pt,
xrightmargin=0pt, style=stHaskell, linerange=\-\-\ bitonic\ sorter-\-\-\ bitonic\ merger, includerangemarker=false, numbers=left, name=numbered, label={code:bit2-bsort}, caption={Parallel bSort}]
import System.Environment(getArgs)
import Control.Parallel.Eden.Auxiliary
import System.Random
import Data.Time.Clock (diffUTCTime, getCurrentTime)
import Data.List
import Data.List.Ordered as Ordered
import Control.Parallel.Eden -- rnf
import Control.Parallel.Eden.Map
import qualified Data.Vector.Unboxed as V

data Direction = Up | Down deriving Show
instance NFData Direction
instance Trans Direction

instance Trans (V.Vector a)

main = do args <- getArgs
          if length args < 2 then print usage else
               do let version = (args!!0) :: String
                  let size = read (args!!1) :: Int
                  let seed = 42
                  let rs = randomlist size seed
                  rnf rs `pseq` putStrLn "list generated"
                  time1 <- getCurrentTime
                  let x = bitonicSort version rs 
                  rnf x `pseq` putStrLn "list sorted"
                  time2 <- getCurrentTime
                  putStrLn $ "Runtime:"++ (show (diffUTCTime time2 time1)) 

usage :: String
usage = ""

randomlist :: Int -> Int -> [Int]
randomlist n seed = take n $ randoms $ mkStdGen seed

-- bitonic sorter 
bSort :: Trans a 
     => (Direction -> [a] -> [a]) -- ^specialized comparison element
     -> Direction                 -- ^sorting direction
     -> [a] -> [a]                -- ^input and ^output 
bSort _         _ [ ] = [ ]   
bSort _         _ [x] = [x] 
bSort sCompElem d xss = (bMerge sCompElem d) . prodBList $ xss where--c \label{line:main} c-- 
    prodBList = unSplit . pMap bSort' . zip [Up, Down] . splitHalf--c \label{line:parMap1} \label{mark:number} c-- 
    bSort' =  uncurry (bSort sCompElem)
    pMap = parMapAt [selfPe, selfPe+hcc] 
    hcc = (length xss) `div` 4 {- half comparator count -} --c \label{line:hno1} c--  

-- bitonic merger 
bMerge :: Trans a 
     => (Direction -> [a] -> [a]) -- ^specialized comparison element
     -> Direction                 -- ^sorting direction 
     -> [a] -> [a]                -- ^input and ^output
bMerge sCompElem d xss@[x,y] = sCompElem d xss
bMerge sCompElem d xss = unSplit . pMap (bMerge sCompElem d) . bSplit $ xss where--c \label{line:parMap2} c-- 
    bSplit = splitHalf . shuffle . pMap' (sCompElem d) . perfectShuffle --c \label{line:parMap3} c-- 
    pMap = parMapAt [selfPe, selfPe+hcc]
    hcc = (length xss) `div` 4 {- half comparator count -} 
    pMap' = parMapAt [selfPe..] 

-- Stones perfect shuffle
perfectShuffle :: [a] -> [[a]] 
perfectShuffle xs = unshuffle halfSize xs
                        where halfSize = (length xs) `div` 2

-- splitHalf
splitHalf :: [a] -> [[a]]
splitHalf = splitIntoN 2

-- simpleMergeSplit
simpleMergeSplit :: Ord a => [a] -> [a] -> ([a],[a])
simpleMergeSplit [] a2 = ([], a2)
simpleMergeSplit a1 [] = (a1, [])
simpleMergeSplit a1 a2 = (s,b) where
    ag = Ordered.merge a1 a2 -- merge from Data.List.Ordered
    lb = max (minimum a1) (minimum a2)
    ub = min (maximum a1) (maximum a2)
    l = [x | x <- ag, x < lb]
    m = [x | x <- ag, x >= lb, x <= ub]
    u = [x | x <- ag, x > ub]
    (m1,m2) = balancingSplit (length u - length l) m
    s = l ++ m1
    b = m2 ++ u
balancingSplit :: Int -> [a] -> ([a],[a])
balancingSplit d xs = splitAt lh xs where
    lh = div ((length xs)+d) 2

-- simpleSort
bitonicSort "simple" = unwrap . bSort sCompElem Up . wrap where 
    wrap = releaseAll 
    unwrap = fetchAll
    sCompElem :: (Trans a, Ord a) => Direction -> [RD a] -> [RD a]
    sCompElem d = releaseAll . compElem d . fetchAll
 
    -- original comparison element
    compElem :: Ord a => Direction -> [a] -> [a]
    compElem Up [x,y] = if x <= y then [x,y] else [y,x]
    compElem Down xs = reverse $ compElem Up xs

-- block 
bitonicSort "block" = unwrap . bSort sCompElem Up . preSort . wrap where
  wrap = releaseAll . map V.fromList . splitIntoN p 
  unwrap = concat . map V.toList . fetchAll
  p = noPe * 2 -- two input lists per row, one row for every PE
  places = 1 : 1 : map (1+) places
  
  preSort :: (V.Unbox a, Ord a) => [RD (V.Vector a)] -> [RD (V.Vector a)] 
  preSort = parMapAt places sSort where
      sSort = release . V.fromList . sort . V.toList . fetch
     
  sCompElem :: (V.Unbox a, Trans a, Ord a) 
              => Direction -> [RD (V.Vector a)] -> [RD (V.Vector a)]
  sCompElem d = releaseAll . map V.fromList . compElemB d . map V.toList . fetchAll    
  
  compElemB :: Ord a => Direction -> [[a]] -> [[a]]
  compElemB Up [xs,ys] = (\(x,y) -> [x,y]) $ simpleMergeSplit xs ys
  compElemB Down xss = reverse $ compElemB Up xss

-- fused 
bitonicSort "fMergeSplit" = unwrap . bSort sCompElem Up . preSort . wrap where
    wrap = releaseAll . map V.fromList . splitIntoN p 
    unwrap = concat . map V.toList . fetchAll

    p = noPe * 2 -- two input lists per row, one row for every PE
    places = 1 : 1 : map (1+) places

    preSort :: (V.Unbox a, Ord a) => [RD (V.Vector a)] -> [RD (V.Vector a)] 
    preSort = parMapAt places sSort

    sSort :: (V.Unbox a, Ord a) => RD (V.Vector a) -> RD (V.Vector a)
    sSort = release . V.fromList . sort . V.toList . fetch
       
    sCompElem :: (V.Unbox a, Trans a, Ord a) 
                => Direction 
                -> [RD (V.Vector a)] -> [RD (V.Vector a)]
    sCompElem d = releaseAll . map V.fromList . compElemB d . map V.toList . fetchAll    

compElemB :: Ord a => Direction -> [[a]] -> [[a]]
compElemB Up [xs,ys] = (\(x,y) -> [x,y]) $ mergeSplit xs ys
compElemB Down xss = reverse $ compElemB Up xss

-- mergeSplit
mergeSplit :: Ord t => [t] -> [t] -> ([t], [t])
mergeSplit [] ys = ([],ys)
mergeSplit xs [] = ([],xs)
mergeSplit xs@(minX:xss) ys@(minY:yss) = if minX <= minY
                            then lBound xs minY ys 0
                            else lBound ys minX xs 0

-- Lower Bound for split
lBound [] _ ys cu = ([],ys) -- nur wenn keine Überlappung existiert
lBound (x:xs) y ys cu
    | x < y = let (u, o) = lBound xs y ys (cu+1)
                  in ((x:u), o)
    | otherwise = let (m,o,on) = oBSplit (x:xs) ys (cu+1) 0 -- Ende von u erreicht
                      in (m, o)

-- ordered balanced split
oBSplit [] ys _ _ = ([],ys,(length ys)) -- upperBound erreicht
oBSplit xs [] _ _ = ([],xs,(length xs)) -- upperBound erreicht
oBSplit (x:xs) (y:ys) un cm
    | x <= y = let (m,o,on) = oBSplit xs (y:ys) un (cm+1)
                   in if cm+un <= on
                        then ((x:m), o, on)
                        else (m,(x:o),on+1)
    | otherwise = let (m,o,on) = oBSplit (x:xs) ys un (cm+1)
                      in if cm+un <= on
                            then ((y:m), o, on)
                            else (m,(y:o),on+1)

\end{lstlisting}
The \cd{bSort} function takes three arguments: an oriented comparison element, a \cd{Direction} denoting whether the result should be sorted in an ascending or descending order and an input list. The main part of the algorithm is a composition of the \cd{prodBList} and the \cd{bMerge} function (cf.\ \autoref{line:main} in \autoref{code:bit2-bsort}). 
The \cd{prodBList} function splits the input list and sorts both parts with the Bitonic Sorter, one half ascending and one half descending (cf.\ \autoref{mark:number}). It uses two helper functions \cd{splitHalf} and \cd{unSplit}. 
With the help of Eden's \cd{splitIntoN}, which splits the input list blockwise into as many parts as the first parameter determines, 
we define:

\begin{code}
splitHalf :: [a] -> [[a]]
splitHalf = splitIntoN 2
\end{code}

Both resulting lists are of the same size because the width of the Bitonic Sorter and therefore its input list's length are powers of two (not to be confused with the size of the blocks which can be of arbitrary size). The needed reverse function -- \cd{unSplit} -- can be defined  as:

\begin{code}
unSplit :: [[a]] -> [a]
unSplit = concat
\end{code}

The correct placement by line is calculated depending on the width of the sorting network (cf.\ \autoref{line:hno1}). Two elements are needed for every comparison element, therefore \cd{hcc} is half the size of the sorting network in the actual recursion step. 
The \cd{bMerge} function does have the same type signature as the \cd{bSort} function but the input list must be a bitonic list for the function to work correctly: 

\begin{lstlisting}[language=Haskell,breaklines=true,xleftmargin=6pt,breakindent=32pt,
xrightmargin=0pt, style=stHaskell, linerange=\-\-\ bitonic\ merger-\-\-\ Stones, includerangemarker=false, numbers=left, name=numbered, label={code:bit2-bsort2}, caption={Parallel bMerge}]
import System.Environment(getArgs)
import Control.Parallel.Eden.Auxiliary
import System.Random
import Data.Time.Clock (diffUTCTime, getCurrentTime)
import Data.List
import Data.List.Ordered as Ordered
import Control.Parallel.Eden -- rnf
import Control.Parallel.Eden.Map
import qualified Data.Vector.Unboxed as V

data Direction = Up | Down deriving Show
instance NFData Direction
instance Trans Direction

instance Trans (V.Vector a)

main = do args <- getArgs
          if length args < 2 then print usage else
               do let version = (args!!0) :: String
                  let size = read (args!!1) :: Int
                  let seed = 42
                  let rs = randomlist size seed
                  rnf rs `pseq` putStrLn "list generated"
                  time1 <- getCurrentTime
                  let x = bitonicSort version rs 
                  rnf x `pseq` putStrLn "list sorted"
                  time2 <- getCurrentTime
                  putStrLn $ "Runtime:"++ (show (diffUTCTime time2 time1)) 

usage :: String
usage = ""

randomlist :: Int -> Int -> [Int]
randomlist n seed = take n $ randoms $ mkStdGen seed

-- bitonic sorter 
bSort :: Trans a 
     => (Direction -> [a] -> [a]) -- ^specialized comparison element
     -> Direction                 -- ^sorting direction
     -> [a] -> [a]                -- ^input and ^output 
bSort _         _ [ ] = [ ]   
bSort _         _ [x] = [x] 
bSort sCompElem d xss = (bMerge sCompElem d) . prodBList $ xss where--c \label{line:main} c-- 
    prodBList = unSplit . pMap bSort' . zip [Up, Down] . splitHalf--c \label{line:parMap1} \label{mark:number} c-- 
    bSort' =  uncurry (bSort sCompElem)
    pMap = parMapAt [selfPe, selfPe+hcc] 
    hcc = (length xss) `div` 4 {- half comparator count -} --c \label{line:hno1} c--  

-- bitonic merger 
bMerge :: Trans a 
     => (Direction -> [a] -> [a]) -- ^specialized comparison element
     -> Direction                 -- ^sorting direction 
     -> [a] -> [a]                -- ^input and ^output
bMerge sCompElem d xss@[x,y] = sCompElem d xss
bMerge sCompElem d xss = unSplit . pMap (bMerge sCompElem d) . bSplit $ xss where--c \label{line:parMap2} c-- 
    bSplit = splitHalf . shuffle . pMap' (sCompElem d) . perfectShuffle --c \label{line:parMap3} c-- 
    pMap = parMapAt [selfPe, selfPe+hcc]
    hcc = (length xss) `div` 4 {- half comparator count -} 
    pMap' = parMapAt [selfPe..] 

-- Stones perfect shuffle
perfectShuffle :: [a] -> [[a]] 
perfectShuffle xs = unshuffle halfSize xs
                        where halfSize = (length xs) `div` 2

-- splitHalf
splitHalf :: [a] -> [[a]]
splitHalf = splitIntoN 2

-- simpleMergeSplit
simpleMergeSplit :: Ord a => [a] -> [a] -> ([a],[a])
simpleMergeSplit [] a2 = ([], a2)
simpleMergeSplit a1 [] = (a1, [])
simpleMergeSplit a1 a2 = (s,b) where
    ag = Ordered.merge a1 a2 -- merge from Data.List.Ordered
    lb = max (minimum a1) (minimum a2)
    ub = min (maximum a1) (maximum a2)
    l = [x | x <- ag, x < lb]
    m = [x | x <- ag, x >= lb, x <= ub]
    u = [x | x <- ag, x > ub]
    (m1,m2) = balancingSplit (length u - length l) m
    s = l ++ m1
    b = m2 ++ u
balancingSplit :: Int -> [a] -> ([a],[a])
balancingSplit d xs = splitAt lh xs where
    lh = div ((length xs)+d) 2

-- simpleSort
bitonicSort "simple" = unwrap . bSort sCompElem Up . wrap where 
    wrap = releaseAll 
    unwrap = fetchAll
    sCompElem :: (Trans a, Ord a) => Direction -> [RD a] -> [RD a]
    sCompElem d = releaseAll . compElem d . fetchAll
 
    -- original comparison element
    compElem :: Ord a => Direction -> [a] -> [a]
    compElem Up [x,y] = if x <= y then [x,y] else [y,x]
    compElem Down xs = reverse $ compElem Up xs

-- block 
bitonicSort "block" = unwrap . bSort sCompElem Up . preSort . wrap where
  wrap = releaseAll . map V.fromList . splitIntoN p 
  unwrap = concat . map V.toList . fetchAll
  p = noPe * 2 -- two input lists per row, one row for every PE
  places = 1 : 1 : map (1+) places
  
  preSort :: (V.Unbox a, Ord a) => [RD (V.Vector a)] -> [RD (V.Vector a)] 
  preSort = parMapAt places sSort where
      sSort = release . V.fromList . sort . V.toList . fetch
     
  sCompElem :: (V.Unbox a, Trans a, Ord a) 
              => Direction -> [RD (V.Vector a)] -> [RD (V.Vector a)]
  sCompElem d = releaseAll . map V.fromList . compElemB d . map V.toList . fetchAll    
  
  compElemB :: Ord a => Direction -> [[a]] -> [[a]]
  compElemB Up [xs,ys] = (\(x,y) -> [x,y]) $ simpleMergeSplit xs ys
  compElemB Down xss = reverse $ compElemB Up xss

-- fused 
bitonicSort "fMergeSplit" = unwrap . bSort sCompElem Up . preSort . wrap where
    wrap = releaseAll . map V.fromList . splitIntoN p 
    unwrap = concat . map V.toList . fetchAll

    p = noPe * 2 -- two input lists per row, one row for every PE
    places = 1 : 1 : map (1+) places

    preSort :: (V.Unbox a, Ord a) => [RD (V.Vector a)] -> [RD (V.Vector a)] 
    preSort = parMapAt places sSort

    sSort :: (V.Unbox a, Ord a) => RD (V.Vector a) -> RD (V.Vector a)
    sSort = release . V.fromList . sort . V.toList . fetch
       
    sCompElem :: (V.Unbox a, Trans a, Ord a) 
                => Direction 
                -> [RD (V.Vector a)] -> [RD (V.Vector a)]
    sCompElem d = releaseAll . map V.fromList . compElemB d . map V.toList . fetchAll    

compElemB :: Ord a => Direction -> [[a]] -> [[a]]
compElemB Up [xs,ys] = (\(x,y) -> [x,y]) $ mergeSplit xs ys
compElemB Down xss = reverse $ compElemB Up xss

-- mergeSplit
mergeSplit :: Ord t => [t] -> [t] -> ([t], [t])
mergeSplit [] ys = ([],ys)
mergeSplit xs [] = ([],xs)
mergeSplit xs@(minX:xss) ys@(minY:yss) = if minX <= minY
                            then lBound xs minY ys 0
                            else lBound ys minX xs 0

-- Lower Bound for split
lBound [] _ ys cu = ([],ys) -- nur wenn keine Überlappung existiert
lBound (x:xs) y ys cu
    | x < y = let (u, o) = lBound xs y ys (cu+1)
                  in ((x:u), o)
    | otherwise = let (m,o,on) = oBSplit (x:xs) ys (cu+1) 0 -- Ende von u erreicht
                      in (m, o)

-- ordered balanced split
oBSplit [] ys _ _ = ([],ys,(length ys)) -- upperBound erreicht
oBSplit xs [] _ _ = ([],xs,(length xs)) -- upperBound erreicht
oBSplit (x:xs) (y:ys) un cm
    | x <= y = let (m,o,on) = oBSplit xs (y:ys) un (cm+1)
                   in if cm+un <= on
                        then ((x:m), o, on)
                        else (m,(x:o),on+1)
    | otherwise = let (m,o,on) = oBSplit (x:xs) ys un (cm+1)
                      in if cm+un <= on
                            then ((y:m), o, on)
                            else (m,(y:o),on+1)

\end{lstlisting}

The main part of the \cd{bMerge} function is the function \cd{bSplit} which splits a bitonic sequence into two bitonic sequences with an order between each other. This function uses a communication structure referred to as a {\it perfect shuffle}\footnote{This structure can be found in various algorithms, e.g.\ in the Fast Fourier transform or in matrix transpositions.} 
by Stone~\cite{stone71}. With this communication scheme the element $i$ and $i+\frac{p}{2}$ are compared, resulting in a split depicted in \autoref{fig:zerlegung}. 

\begin{figure}[htb]
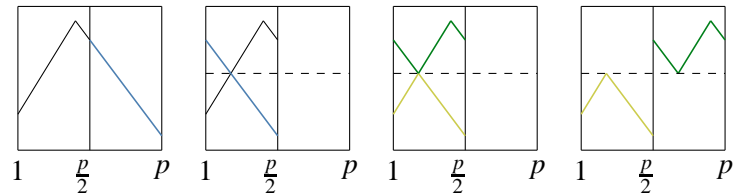

    \centering
    \begin{subfigure}[b]{0.15\textwidth}
        \centering
        \includegraphics[width=0.99\linewidth]{figures/zerlegung_01.tikz}
    \end{subfigure}
    \begin{subfigure}[b]{0.15\textwidth}
        \centering
        \includegraphics[width=0.99\linewidth]{figures/zerlegung_02.tikz}
    \end{subfigure}
    \begin{subfigure}[b]{0.15\textwidth}
        \centering
        \includegraphics[width=0.99\linewidth]{figures/zerlegung_03.tikz}
    \end{subfigure}
    \begin{subfigure}[b]{0.15\textwidth}
        \centering
        \includegraphics[width=0.99\linewidth]{figures/zerlegung_04.tikz}
    \end{subfigure}
    \caption{Concept of splitting a bitonic sequence.}
    \label{fig:zerlegung}
\end{figure}

In Eden this {\it perfect shuffle} is easily defined with the help of the offered auxiliary functions: 
\begin{code}
-- Round robin distribution and inverse function called shuffle 
unshuffle :: Int -> [a] -> [[a]]
shuffle :: [[a]] -> [a]
\end{code}

The first parameter of \cd{unshuffle} specifies the number of sublists in which the list is split, e.g.:
\begin{code}
unshuffle 3 [1..10] = [[1,4,7,10],[2,5,8],[3,6,9]]
shuffle [[1,4,7,10],[2,5,8],[3,6,9]] = [1..10]
\end{code}

The {\it perfect shuffle} is then defined as:
\begin{lstlisting}[language=Haskell,breaklines=true,xleftmargin=0pt,
      xrightmargin=0pt, style=stHaskell, linerange=\-\-\ Stones-\-\-\ , includerangemarker=false]
import System.Environment(getArgs)
import Control.Parallel.Eden.Auxiliary
import System.Random
import Data.Time.Clock (diffUTCTime, getCurrentTime)

import Data.List
import Data.Tuple

import Control.Parallel.Eden --rnf

main = do
    args <- getArgs
    let size = read (args!!0) :: Int
    let seed = 42
    let rs = randomlist size seed
    rnf rs `pseq` putStr "" -- "Liste generiert"
    time1 <- getCurrentTime
    let x = bitonicSort rs
    rnf x `pseq` putStr "" -- "Liste sortiert"
    time2 <- getCurrentTime
    --putStrLn $ (show size)++","++
    putStrLn $ (show (diffUTCTime time2 time1)) 
 
-- generating list of random numbers 
randomlist :: Int -> Int -> [Int]
randomlist n seed = take n $ randoms $ mkStdGen seed


-- Aufruf 
bitonicSort :: (Ord a) => [a] -> [a]
bitonicSort = bSort Up 

-- Sortierer
bSort :: Ord a => Direction -> [a] -> [a] -- for lists of length 2^k
bSort _ []  = []   
bSort _ [x] = [x]  
bSort d xss = (bMerge d) . prodBList $ xss where 
    prodBList :: Ord a => [a] -> [a]
    prodBList = unSplit . map (uncurry bSort) . zip [Up, Down] . splitHalf 

-- Mischer
bMerge :: (Ord a) => Direction -> [a] -> [a] -- for lists of length 2^k 
bMerge d [x,y] = compElem d (x,y)
bMerge d xss = unSplit . map (bMerge d) . bSplit $ xss
    where bSplit :: Ord a => [a] -> [[a]]
          bSplit = splitHalf . shuffle . map (compElem d) . perfectShuffle

-- Komparatorelement
data Direction = Up | Down deriving Show

compElem :: Ord a => Direction -> [a] -> [a]
compElem Up [x,y] = if x <= y then [x,y] else [y,x]
compElem Down xs = reverse $ compElem Up xs

-- Hilfsfunktionen
splitHalf :: [a] -> [[a]]
splitHalf = splitIntoN 2

-- Stones perfect shuffle
perfectShuffle :: [a] -> [[a]] 
perfectShuffle xs = unshuffle halfSize xs
                        where halfSize = (length xs) `div` 2
\end{lstlisting}

A direct communication between consecutive comparison elements can be realised with Eden's Remote Data concept in which a smaller handle is transmitted instead of the actual data. 
The data itself is fetched directly when needed from the PE where the handle was created.  
The more intermediate steps are involved, the more effective the benefits of this concept become.
This can be done by the operations: % \cd{fetch} and \cd{release}:
\begin{code}
type RD a -- remote data
-- converts local data into corresponding remote data and vice versa 
release :: Trans a => a -> RD a
fetch :: Trans a => RD a -> a
releaseAll :: Trans a => [a] -> [RD a] -- list variants
fetchAll :: Trans a => [RD a] -> [a]
\end{code}
In \autoref{fig:rd} the communication scheme of a Remote Data connection is pictured.
\begin{figure}[h]
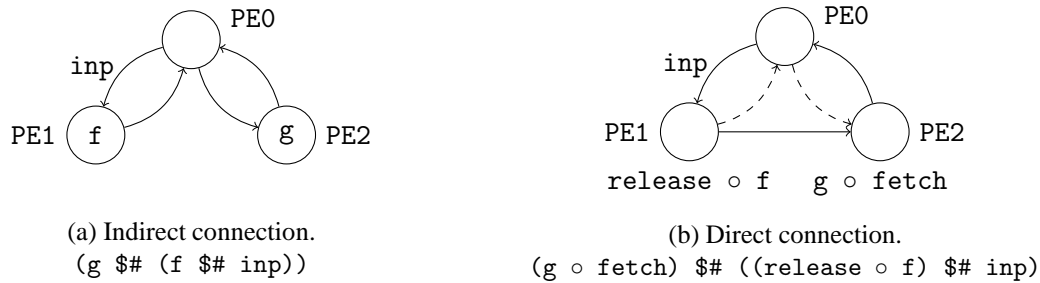

    \centering
    \begin{subfigure}{0.45\textwidth}
        \captionsetup{justification=centering}
        \centering
        \includegraphics[width=.7\linewidth]{figures/rd_01.tikz}
        \caption[Indirekte Verbindung]{Indirect connection. \\ \protect \cd{(g \$\# (f \$\# inp))}}
        \label{fig:rda}
    \end{subfigure}
    \hspace{0.5cm}
    \begin{subfigure}{0.45\textwidth}
        \centering
        \captionsetup{justification=centering}
        \includegraphics[width=.7\linewidth]{figures/rd_02.tikz}
        \caption{Direct connection. \\ \protect \cd{(g . fetch) \$\# ((release . f) \$\# inp)}}
        \label{fig:rdb}
    \end{subfigure}
    \caption[Remote Data Scheme]{Remote Data scheme. Source: \protect \cite{loogen11}. The processes computing the results of \cd{f} and \cd{g} are placed on two different \cd{PEs}. Without RD, the result of \cd{f} is transferred via the parental process. %(on \cd{PE0}).
    With RD a handle is generated on \cd{PE1} and transferred via \cd{PE0} to \cd{PE2}, the actual result is transferred directly. }
    \label{fig:rd}
\end{figure}

If we call the \cd{bSort} function with the original comparison element and the organization of the communication via Remote Data we receive a correct implementation of the Bitonic Sorter:  

\begin{lstlisting}[language=Haskell,breaklines=true,xleftmargin=6pt,breakindent=32pt,
xrightmargin=0pt, style=stHaskell, includerangemarker=false, linerange=\-\-\ simpleSort-\-\-\ original, numbers=left, name=numbered, label={code:bsneu}, caption={Parallel variant of the original Bitonic Sorter}]
import System.Environment(getArgs)
import Control.Parallel.Eden.Auxiliary
import System.Random
import Data.Time.Clock (diffUTCTime, getCurrentTime)
import Data.List
import Data.List.Ordered as Ordered
import Control.Parallel.Eden -- rnf
import Control.Parallel.Eden.Map
import qualified Data.Vector.Unboxed as V

data Direction = Up | Down deriving Show
instance NFData Direction
instance Trans Direction

instance Trans (V.Vector a)

main = do args <- getArgs
          if length args < 2 then print usage else
               do let version = (args!!0) :: String
                  let size = read (args!!1) :: Int
                  let seed = 42
                  let rs = randomlist size seed
                  rnf rs `pseq` putStrLn "list generated"
                  time1 <- getCurrentTime
                  let x = bitonicSort version rs 
                  rnf x `pseq` putStrLn "list sorted"
                  time2 <- getCurrentTime
                  putStrLn $ "Runtime:"++ (show (diffUTCTime time2 time1)) 

usage :: String
usage = ""

randomlist :: Int -> Int -> [Int]
randomlist n seed = take n $ randoms $ mkStdGen seed

-- bitonic sorter 
bSort :: Trans a 
     => (Direction -> [a] -> [a]) -- ^specialized comparison element
     -> Direction                 -- ^sorting direction
     -> [a] -> [a]                -- ^input and ^output 
bSort _         _ [ ] = [ ]   
bSort _         _ [x] = [x] 
bSort sCompElem d xss = (bMerge sCompElem d) . prodBList $ xss where--c \label{line:main} c-- 
    prodBList = unSplit . pMap bSort' . zip [Up, Down] . splitHalf--c \label{line:parMap1} \label{mark:number} c-- 
    bSort' =  uncurry (bSort sCompElem)
    pMap = parMapAt [selfPe, selfPe+hcc] 
    hcc = (length xss) `div` 4 {- half comparator count -} --c \label{line:hno1} c--  

-- bitonic merger 
bMerge :: Trans a 
     => (Direction -> [a] -> [a]) -- ^specialized comparison element
     -> Direction                 -- ^sorting direction 
     -> [a] -> [a]                -- ^input and ^output
bMerge sCompElem d xss@[x,y] = sCompElem d xss
bMerge sCompElem d xss = unSplit . pMap (bMerge sCompElem d) . bSplit $ xss where--c \label{line:parMap2} c-- 
    bSplit = splitHalf . shuffle . pMap' (sCompElem d) . perfectShuffle --c \label{line:parMap3} c-- 
    pMap = parMapAt [selfPe, selfPe+hcc]
    hcc = (length xss) `div` 4 {- half comparator count -} 
    pMap' = parMapAt [selfPe..] 

-- Stones perfect shuffle
perfectShuffle :: [a] -> [[a]] 
perfectShuffle xs = unshuffle halfSize xs
                        where halfSize = (length xs) `div` 2

-- splitHalf
splitHalf :: [a] -> [[a]]
splitHalf = splitIntoN 2

-- simpleMergeSplit
simpleMergeSplit :: Ord a => [a] -> [a] -> ([a],[a])
simpleMergeSplit [] a2 = ([], a2)
simpleMergeSplit a1 [] = (a1, [])
simpleMergeSplit a1 a2 = (s,b) where
    ag = Ordered.merge a1 a2 -- merge from Data.List.Ordered
    lb = max (minimum a1) (minimum a2)
    ub = min (maximum a1) (maximum a2)
    l = [x | x <- ag, x < lb]
    m = [x | x <- ag, x >= lb, x <= ub]
    u = [x | x <- ag, x > ub]
    (m1,m2) = balancingSplit (length u - length l) m
    s = l ++ m1
    b = m2 ++ u
balancingSplit :: Int -> [a] -> ([a],[a])
balancingSplit d xs = splitAt lh xs where
    lh = div ((length xs)+d) 2

-- simpleSort
bitonicSort "simple" = unwrap . bSort sCompElem Up . wrap where 
    wrap = releaseAll 
    unwrap = fetchAll
    sCompElem :: (Trans a, Ord a) => Direction -> [RD a] -> [RD a]
    sCompElem d = releaseAll . compElem d . fetchAll
 
    -- original comparison element
    compElem :: Ord a => Direction -> [a] -> [a]
    compElem Up [x,y] = if x <= y then [x,y] else [y,x]
    compElem Down xs = reverse $ compElem Up xs

-- block 
bitonicSort "block" = unwrap . bSort sCompElem Up . preSort . wrap where
  wrap = releaseAll . map V.fromList . splitIntoN p 
  unwrap = concat . map V.toList . fetchAll
  p = noPe * 2 -- two input lists per row, one row for every PE
  places = 1 : 1 : map (1+) places
  
  preSort :: (V.Unbox a, Ord a) => [RD (V.Vector a)] -> [RD (V.Vector a)] 
  preSort = parMapAt places sSort where
      sSort = release . V.fromList . sort . V.toList . fetch
     
  sCompElem :: (V.Unbox a, Trans a, Ord a) 
              => Direction -> [RD (V.Vector a)] -> [RD (V.Vector a)]
  sCompElem d = releaseAll . map V.fromList . compElemB d . map V.toList . fetchAll    
  
  compElemB :: Ord a => Direction -> [[a]] -> [[a]]
  compElemB Up [xs,ys] = (\(x,y) -> [x,y]) $ simpleMergeSplit xs ys
  compElemB Down xss = reverse $ compElemB Up xss

-- fused 
bitonicSort "fMergeSplit" = unwrap . bSort sCompElem Up . preSort . wrap where
    wrap = releaseAll . map V.fromList . splitIntoN p 
    unwrap = concat . map V.toList . fetchAll

    p = noPe * 2 -- two input lists per row, one row for every PE
    places = 1 : 1 : map (1+) places

    preSort :: (V.Unbox a, Ord a) => [RD (V.Vector a)] -> [RD (V.Vector a)] 
    preSort = parMapAt places sSort

    sSort :: (V.Unbox a, Ord a) => RD (V.Vector a) -> RD (V.Vector a)
    sSort = release . V.fromList . sort . V.toList . fetch
       
    sCompElem :: (V.Unbox a, Trans a, Ord a) 
                => Direction 
                -> [RD (V.Vector a)] -> [RD (V.Vector a)]
    sCompElem d = releaseAll . map V.fromList . compElemB d . map V.toList . fetchAll    

compElemB :: Ord a => Direction -> [[a]] -> [[a]]
compElemB Up [xs,ys] = (\(x,y) -> [x,y]) $ mergeSplit xs ys
compElemB Down xss = reverse $ compElemB Up xss

-- mergeSplit
mergeSplit :: Ord t => [t] -> [t] -> ([t], [t])
mergeSplit [] ys = ([],ys)
mergeSplit xs [] = ([],xs)
mergeSplit xs@(minX:xss) ys@(minY:yss) = if minX <= minY
                            then lBound xs minY ys 0
                            else lBound ys minX xs 0

-- Lower Bound for split
lBound [] _ ys cu = ([],ys) -- nur wenn keine Überlappung existiert
lBound (x:xs) y ys cu
    | x < y = let (u, o) = lBound xs y ys (cu+1)
                  in ((x:u), o)
    | otherwise = let (m,o,on) = oBSplit (x:xs) ys (cu+1) 0 -- Ende von u erreicht
                      in (m, o)

-- ordered balanced split
oBSplit [] ys _ _ = ([],ys,(length ys)) -- upperBound erreicht
oBSplit xs [] _ _ = ([],xs,(length xs)) -- upperBound erreicht
oBSplit (x:xs) (y:ys) un cm
    | x <= y = let (m,o,on) = oBSplit xs (y:ys) un (cm+1)
                   in if cm+un <= on
                        then ((x:m), o, on)
                        else (m,(x:o),on+1)
    | otherwise = let (m,o,on) = oBSplit (x:xs) ys un (cm+1)
                      in if cm+un <= on
                            then ((y:m), o, on)
                            else (m,(y:o),on+1)

\end{lstlisting}

To apply the agglomeration law we can change the comparison element to a suitable comparison element for blocks. In \autoref{code:sce} a simple implementation is given: 

\begin{lstlisting}[language=Haskell, breaklines=true,xleftmargin=6pt,
    xrightmargin=0pt, style=stHaskell, numbers=left, linerange=\-\-\ simpleMergeSplit-\-\-\ simpleSort, includerangemarker=false, label={code:sce}, caption={A simple MergeSplit function working as a comparison element for blocks of data}]
import System.Environment(getArgs)
import Control.Parallel.Eden.Auxiliary
import System.Random
import Data.Time.Clock (diffUTCTime, getCurrentTime)
import Data.List
import Data.List.Ordered as Ordered
import Control.Parallel.Eden -- rnf
import Control.Parallel.Eden.Map
import qualified Data.Vector.Unboxed as V

data Direction = Up | Down deriving Show
instance NFData Direction
instance Trans Direction

instance Trans (V.Vector a)

main = do args <- getArgs
          if length args < 2 then print usage else
               do let version = (args!!0) :: String
                  let size = read (args!!1) :: Int
                  let seed = 42
                  let rs = randomlist size seed
                  rnf rs `pseq` putStrLn "list generated"
                  time1 <- getCurrentTime
                  let x = bitonicSort version rs 
                  rnf x `pseq` putStrLn "list sorted"
                  time2 <- getCurrentTime
                  putStrLn $ "Runtime:"++ (show (diffUTCTime time2 time1)) 

usage :: String
usage = ""

randomlist :: Int -> Int -> [Int]
randomlist n seed = take n $ randoms $ mkStdGen seed

-- bitonic sorter 
bSort :: Trans a 
     => (Direction -> [a] -> [a]) -- ^specialized comparison element
     -> Direction                 -- ^sorting direction
     -> [a] -> [a]                -- ^input and ^output 
bSort _         _ [ ] = [ ]   
bSort _         _ [x] = [x] 
bSort sCompElem d xss = (bMerge sCompElem d) . prodBList $ xss where--c \label{line:main} c-- 
    prodBList = unSplit . pMap bSort' . zip [Up, Down] . splitHalf--c \label{line:parMap1} \label{mark:number} c-- 
    bSort' =  uncurry (bSort sCompElem)
    pMap = parMapAt [selfPe, selfPe+hcc] 
    hcc = (length xss) `div` 4 {- half comparator count -} --c \label{line:hno1} c--  

-- bitonic merger 
bMerge :: Trans a 
     => (Direction -> [a] -> [a]) -- ^specialized comparison element
     -> Direction                 -- ^sorting direction 
     -> [a] -> [a]                -- ^input and ^output
bMerge sCompElem d xss@[x,y] = sCompElem d xss
bMerge sCompElem d xss = unSplit . pMap (bMerge sCompElem d) . bSplit $ xss where--c \label{line:parMap2} c-- 
    bSplit = splitHalf . shuffle . pMap' (sCompElem d) . perfectShuffle --c \label{line:parMap3} c-- 
    pMap = parMapAt [selfPe, selfPe+hcc]
    hcc = (length xss) `div` 4 {- half comparator count -} 
    pMap' = parMapAt [selfPe..] 

-- Stones perfect shuffle
perfectShuffle :: [a] -> [[a]] 
perfectShuffle xs = unshuffle halfSize xs
                        where halfSize = (length xs) `div` 2

-- splitHalf
splitHalf :: [a] -> [[a]]
splitHalf = splitIntoN 2

-- simpleMergeSplit
simpleMergeSplit :: Ord a => [a] -> [a] -> ([a],[a])
simpleMergeSplit [] a2 = ([], a2)
simpleMergeSplit a1 [] = (a1, [])
simpleMergeSplit a1 a2 = (s,b) where
    ag = Ordered.merge a1 a2 -- merge from Data.List.Ordered
    lb = max (minimum a1) (minimum a2)
    ub = min (maximum a1) (maximum a2)
    l = [x | x <- ag, x < lb]
    m = [x | x <- ag, x >= lb, x <= ub]
    u = [x | x <- ag, x > ub]
    (m1,m2) = balancingSplit (length u - length l) m
    s = l ++ m1
    b = m2 ++ u
balancingSplit :: Int -> [a] -> ([a],[a])
balancingSplit d xs = splitAt lh xs where
    lh = div ((length xs)+d) 2

-- simpleSort
bitonicSort "simple" = unwrap . bSort sCompElem Up . wrap where 
    wrap = releaseAll 
    unwrap = fetchAll
    sCompElem :: (Trans a, Ord a) => Direction -> [RD a] -> [RD a]
    sCompElem d = releaseAll . compElem d . fetchAll
 
    -- original comparison element
    compElem :: Ord a => Direction -> [a] -> [a]
    compElem Up [x,y] = if x <= y then [x,y] else [y,x]
    compElem Down xs = reverse $ compElem Up xs

-- block 
bitonicSort "block" = unwrap . bSort sCompElem Up . preSort . wrap where
  wrap = releaseAll . map V.fromList . splitIntoN p 
  unwrap = concat . map V.toList . fetchAll
  p = noPe * 2 -- two input lists per row, one row for every PE
  places = 1 : 1 : map (1+) places
  
  preSort :: (V.Unbox a, Ord a) => [RD (V.Vector a)] -> [RD (V.Vector a)] 
  preSort = parMapAt places sSort where
      sSort = release . V.fromList . sort . V.toList . fetch
     
  sCompElem :: (V.Unbox a, Trans a, Ord a) 
              => Direction -> [RD (V.Vector a)] -> [RD (V.Vector a)]
  sCompElem d = releaseAll . map V.fromList . compElemB d . map V.toList . fetchAll    
  
  compElemB :: Ord a => Direction -> [[a]] -> [[a]]
  compElemB Up [xs,ys] = (\(x,y) -> [x,y]) $ simpleMergeSplit xs ys
  compElemB Down xss = reverse $ compElemB Up xss

-- fused 
bitonicSort "fMergeSplit" = unwrap . bSort sCompElem Up . preSort . wrap where
    wrap = releaseAll . map V.fromList . splitIntoN p 
    unwrap = concat . map V.toList . fetchAll

    p = noPe * 2 -- two input lists per row, one row for every PE
    places = 1 : 1 : map (1+) places

    preSort :: (V.Unbox a, Ord a) => [RD (V.Vector a)] -> [RD (V.Vector a)] 
    preSort = parMapAt places sSort

    sSort :: (V.Unbox a, Ord a) => RD (V.Vector a) -> RD (V.Vector a)
    sSort = release . V.fromList . sort . V.toList . fetch
       
    sCompElem :: (V.Unbox a, Trans a, Ord a) 
                => Direction 
                -> [RD (V.Vector a)] -> [RD (V.Vector a)]
    sCompElem d = releaseAll . map V.fromList . compElemB d . map V.toList . fetchAll    

compElemB :: Ord a => Direction -> [[a]] -> [[a]]
compElemB Up [xs,ys] = (\(x,y) -> [x,y]) $ mergeSplit xs ys
compElemB Down xss = reverse $ compElemB Up xss

-- mergeSplit
mergeSplit :: Ord t => [t] -> [t] -> ([t], [t])
mergeSplit [] ys = ([],ys)
mergeSplit xs [] = ([],xs)
mergeSplit xs@(minX:xss) ys@(minY:yss) = if minX <= minY
                            then lBound xs minY ys 0
                            else lBound ys minX xs 0

-- Lower Bound for split
lBound [] _ ys cu = ([],ys) -- nur wenn keine Überlappung existiert
lBound (x:xs) y ys cu
    | x < y = let (u, o) = lBound xs y ys (cu+1)
                  in ((x:u), o)
    | otherwise = let (m,o,on) = oBSplit (x:xs) ys (cu+1) 0 -- Ende von u erreicht
                      in (m, o)

-- ordered balanced split
oBSplit [] ys _ _ = ([],ys,(length ys)) -- upperBound erreicht
oBSplit xs [] _ _ = ([],xs,(length xs)) -- upperBound erreicht
oBSplit (x:xs) (y:ys) un cm
    | x <= y = let (m,o,on) = oBSplit xs (y:ys) un (cm+1)
                   in if cm+un <= on
                        then ((x:m), o, on)
                        else (m,(x:o),on+1)
    | otherwise = let (m,o,on) = oBSplit (x:xs) ys un (cm+1)
                      in if cm+un <= on
                            then ((y:m), o, on)
                            else (m,(y:o),on+1)

\end{lstlisting}

In \autoref{code:bsblock} an optimised implementation is given. It uses unboxed vectors to optimise transmissions. For reasons of comparability, we use the list variant of the altered comparison element and the merge sort from \cd{Data.List}. 

\begin{lstlisting}[language=Haskell,breaklines=true,xleftmargin=6pt,breakindent=32pt,
xrightmargin=0pt, style=stHaskell, includerangemarker=false, linerange=\-\-\ block-\-\-\ fused, numbers=left, name=numbered, label={code:bsblock}, caption={Combination of mergesort with the Bitonic Sorter}]
import System.Environment(getArgs)
import Control.Parallel.Eden.Auxiliary
import System.Random
import Data.Time.Clock (diffUTCTime, getCurrentTime)
import Data.List
import Data.List.Ordered as Ordered
import Control.Parallel.Eden -- rnf
import Control.Parallel.Eden.Map
import qualified Data.Vector.Unboxed as V

data Direction = Up | Down deriving Show
instance NFData Direction
instance Trans Direction

instance Trans (V.Vector a)

main = do args <- getArgs
          if length args < 2 then print usage else
               do let version = (args!!0) :: String
                  let size = read (args!!1) :: Int
                  let seed = 42
                  let rs = randomlist size seed
                  rnf rs `pseq` putStrLn "list generated"
                  time1 <- getCurrentTime
                  let x = bitonicSort version rs 
                  rnf x `pseq` putStrLn "list sorted"
                  time2 <- getCurrentTime
                  putStrLn $ "Runtime:"++ (show (diffUTCTime time2 time1)) 

usage :: String
usage = ""

randomlist :: Int -> Int -> [Int]
randomlist n seed = take n $ randoms $ mkStdGen seed

-- bitonic sorter 
bSort :: Trans a 
     => (Direction -> [a] -> [a]) -- ^specialized comparison element
     -> Direction                 -- ^sorting direction
     -> [a] -> [a]                -- ^input and ^output 
bSort _         _ [ ] = [ ]   
bSort _         _ [x] = [x] 
bSort sCompElem d xss = (bMerge sCompElem d) . prodBList $ xss where--c \label{line:main} c-- 
    prodBList = unSplit . pMap bSort' . zip [Up, Down] . splitHalf--c \label{line:parMap1} \label{mark:number} c-- 
    bSort' =  uncurry (bSort sCompElem)
    pMap = parMapAt [selfPe, selfPe+hcc] 
    hcc = (length xss) `div` 4 {- half comparator count -} --c \label{line:hno1} c--  

-- bitonic merger 
bMerge :: Trans a 
     => (Direction -> [a] -> [a]) -- ^specialized comparison element
     -> Direction                 -- ^sorting direction 
     -> [a] -> [a]                -- ^input and ^output
bMerge sCompElem d xss@[x,y] = sCompElem d xss
bMerge sCompElem d xss = unSplit . pMap (bMerge sCompElem d) . bSplit $ xss where--c \label{line:parMap2} c-- 
    bSplit = splitHalf . shuffle . pMap' (sCompElem d) . perfectShuffle --c \label{line:parMap3} c-- 
    pMap = parMapAt [selfPe, selfPe+hcc]
    hcc = (length xss) `div` 4 {- half comparator count -} 
    pMap' = parMapAt [selfPe..] 

-- Stones perfect shuffle
perfectShuffle :: [a] -> [[a]] 
perfectShuffle xs = unshuffle halfSize xs
                        where halfSize = (length xs) `div` 2

-- splitHalf
splitHalf :: [a] -> [[a]]
splitHalf = splitIntoN 2

-- simpleMergeSplit
simpleMergeSplit :: Ord a => [a] -> [a] -> ([a],[a])
simpleMergeSplit [] a2 = ([], a2)
simpleMergeSplit a1 [] = (a1, [])
simpleMergeSplit a1 a2 = (s,b) where
    ag = Ordered.merge a1 a2 -- merge from Data.List.Ordered
    lb = max (minimum a1) (minimum a2)
    ub = min (maximum a1) (maximum a2)
    l = [x | x <- ag, x < lb]
    m = [x | x <- ag, x >= lb, x <= ub]
    u = [x | x <- ag, x > ub]
    (m1,m2) = balancingSplit (length u - length l) m
    s = l ++ m1
    b = m2 ++ u
balancingSplit :: Int -> [a] -> ([a],[a])
balancingSplit d xs = splitAt lh xs where
    lh = div ((length xs)+d) 2

-- simpleSort
bitonicSort "simple" = unwrap . bSort sCompElem Up . wrap where 
    wrap = releaseAll 
    unwrap = fetchAll
    sCompElem :: (Trans a, Ord a) => Direction -> [RD a] -> [RD a]
    sCompElem d = releaseAll . compElem d . fetchAll
 
    -- original comparison element
    compElem :: Ord a => Direction -> [a] -> [a]
    compElem Up [x,y] = if x <= y then [x,y] else [y,x]
    compElem Down xs = reverse $ compElem Up xs

-- block 
bitonicSort "block" = unwrap . bSort sCompElem Up . preSort . wrap where
  wrap = releaseAll . map V.fromList . splitIntoN p 
  unwrap = concat . map V.toList . fetchAll
  p = noPe * 2 -- two input lists per row, one row for every PE
  places = 1 : 1 : map (1+) places
  
  preSort :: (V.Unbox a, Ord a) => [RD (V.Vector a)] -> [RD (V.Vector a)] 
  preSort = parMapAt places sSort where
      sSort = release . V.fromList . sort . V.toList . fetch
     
  sCompElem :: (V.Unbox a, Trans a, Ord a) 
              => Direction -> [RD (V.Vector a)] -> [RD (V.Vector a)]
  sCompElem d = releaseAll . map V.fromList . compElemB d . map V.toList . fetchAll    
  
  compElemB :: Ord a => Direction -> [[a]] -> [[a]]
  compElemB Up [xs,ys] = (\(x,y) -> [x,y]) $ simpleMergeSplit xs ys
  compElemB Down xss = reverse $ compElemB Up xss

-- fused 
bitonicSort "fMergeSplit" = unwrap . bSort sCompElem Up . preSort . wrap where
    wrap = releaseAll . map V.fromList . splitIntoN p 
    unwrap = concat . map V.toList . fetchAll

    p = noPe * 2 -- two input lists per row, one row for every PE
    places = 1 : 1 : map (1+) places

    preSort :: (V.Unbox a, Ord a) => [RD (V.Vector a)] -> [RD (V.Vector a)] 
    preSort = parMapAt places sSort

    sSort :: (V.Unbox a, Ord a) => RD (V.Vector a) -> RD (V.Vector a)
    sSort = release . V.fromList . sort . V.toList . fetch
       
    sCompElem :: (V.Unbox a, Trans a, Ord a) 
                => Direction 
                -> [RD (V.Vector a)] -> [RD (V.Vector a)]
    sCompElem d = releaseAll . map V.fromList . compElemB d . map V.toList . fetchAll    

compElemB :: Ord a => Direction -> [[a]] -> [[a]]
compElemB Up [xs,ys] = (\(x,y) -> [x,y]) $ mergeSplit xs ys
compElemB Down xss = reverse $ compElemB Up xss

-- mergeSplit
mergeSplit :: Ord t => [t] -> [t] -> ([t], [t])
mergeSplit [] ys = ([],ys)
mergeSplit xs [] = ([],xs)
mergeSplit xs@(minX:xss) ys@(minY:yss) = if minX <= minY
                            then lBound xs minY ys 0
                            else lBound ys minX xs 0

-- Lower Bound for split
lBound [] _ ys cu = ([],ys) -- nur wenn keine Überlappung existiert
lBound (x:xs) y ys cu
    | x < y = let (u, o) = lBound xs y ys (cu+1)
                  in ((x:u), o)
    | otherwise = let (m,o,on) = oBSplit (x:xs) ys (cu+1) 0 -- Ende von u erreicht
                      in (m, o)

-- ordered balanced split
oBSplit [] ys _ _ = ([],ys,(length ys)) -- upperBound erreicht
oBSplit xs [] _ _ = ([],xs,(length xs)) -- upperBound erreicht
oBSplit (x:xs) (y:ys) un cm
    | x <= y = let (m,o,on) = oBSplit xs (y:ys) un (cm+1)
                   in if cm+un <= on
                        then ((x:m), o, on)
                        else (m,(x:o),on+1)
    | otherwise = let (m,o,on) = oBSplit (x:xs) ys un (cm+1)
                      in if cm+un <= on
                            then ((y:m), o, on)
                            else (m,(y:o),on+1)

\end{lstlisting}
This simple adaption results in a hybrid sorting algorithm parallelising merge sort with the Bitonic Sorter. 

\subsection{Runtime and Speedup Evaluation}
We tested the algorithms on a multicore computer called Hex and on the Beowulf Cluster\footnote{Hex is equipped with an {\ttfamily AMD Opteron CPU 6378} (64\,cores) and 64\,GB memory, the Beowulf cluster at the Heriot-Watt University Edinburgh consists of 32 nodes, each one equipped with an {\ttfamily Intel Xeon E5504} CPU (8\,Cores) and 12\,GB memory.}
in order to compare two different implementations of Eden: with MPI~\cite{mpi} as a middleware on the Beowulf cluster and an optimised implementation on the multicore computer.
First we will compare the above-mentioned parallelization of merge sort  using the Bitonic Sorter to another parallelization of the same merge sort using the disDC divide-and-conquer skeleton from Eden's skeleton library. Both variants are implemented in Eden and equipped with similar improvements. We will work on lists in particular since they are the common choice of data structure in Haskell but use unboxed vectors for transmissions. 
In \autoref{fig:hex_time1} the runtime graphs of the parallel disDC merge sort and the Bitonic Sorter are depicted. 
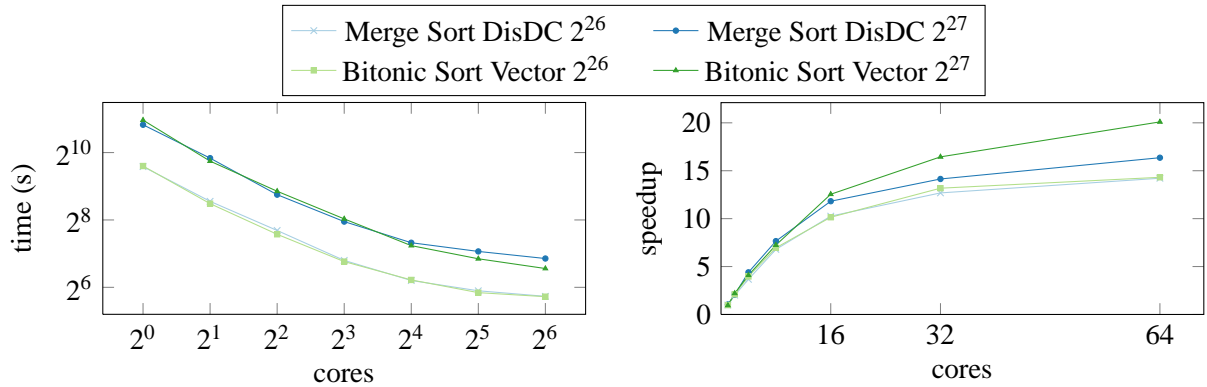
\begin{figure}[htbp]
    \centering
    \begin{tikzpicture}
        \begin{groupplot}[group style={group size=2 by 1, horizontal sep=1.8cm}]
            \nextgroupplot[
                cycle list name=colorlist,
                legend style={align=left,anchor=south,at={(1.1,1.03)},/tikz/every even column/.append style={column sep=0.5cm}},
                legend columns=2,
                ylabel style={align=left},
                width=0.5\textwidth,
                height=0.275\textwidth,
                xmode=log,
                log basis x={2},
                ymode=log,
                log basis y={2},
                xtick={1,2,4,8,16,32,64,128,256},
                xlabel={cores},
                ylabel={time (s)}
            ]
            \addplot table {./graphs/mergesort_min_time_disDCVecChunk_67108864_1000.data};
            \addlegendentry{Merge Sort DisDC $2^{26}$}
            \addplot table {./graphs/mergesort_min_time_disDCVecChunk_134217728_1000.data};
            \addlegendentry{Merge Sort DisDC $2^{27}$}
            \addplot table {./graphs/bitonic-skeleton_min_time_vector_67108864-.data};
            \addlegendentry{Bitonic Sort Vector $2^{26}$}
            \addplot table {./graphs/bitonic-skeleton_min_time_vector_134217728-.data};
            \addlegendentry{Bitonic Sort Vector $2^{27}$};
            \nextgroupplot[
                ymin=0,xmin=0,
                cycle list name=colorlist,
                legend style={align=left,anchor=west},
                ylabel style={align=left},
                width=0.5\textwidth,
                height=0.275\textwidth,
                xtick={16,32,64,128,256},
                xlabel={cores},
                ylabel={speedup}
            ]
            \addplot table {./graphs/mergesort_abs_speedup_min_time_disDCVecChunk_67108864_1000.data};
            \addplot table {./graphs/mergesort_abs_speedup_min_time_disDCVecChunk_134217728_1000.data};
            \addplot table {./graphs/bitonic-skeleton_abs_speedup_min_time_vector_67108864-.data};
            \addplot table {./graphs/bitonic-skeleton_abs_speedup_min_time_vector_134217728-.data};
        \end{groupplot}
    \end{tikzpicture}
    \caption{Runtime and Speedup of the Bitonic Sorter and merge sort on Hex with $2^{26}$ and $2^{27}$ elements.}
    \label{fig:hex_time1}
\end{figure}

The graphs indicate that although the respective runtimes are fairly similar, the Bitonic Sorter variant scales better for larger inputs. The assumption can be hardened by the examination of the corresponding (absolute) speedups.  
The better scalability of the Bitonic Sorter can partly be explained by the merging, consisting of many small steps with comparison elements. This concept of merging can benefit from a great number of \cd{PEs}.  
A discovery that can also be made on the Beowulf Cluster though it is notable that here the perceived characteristics are even more pronounced (cf.\ \autoref{fig:beowulf_time1}). 
\begin{figure}[htbp]
    \centering
    \begin{tikzpicture}
        \begin{groupplot}[group style={group size=2 by 1, horizontal sep=1.8cm}]
            \nextgroupplot[
                cycle list name=colorlist,
                legend style={align=left,anchor=south,at={(1.1,1.03)},/tikz/every even column/.append style={column sep=0.5cm}},
                legend columns=2,
                ylabel style={align=left},
                width=0.5\textwidth,
                height=0.275\textwidth,
                xmode=log,
                log basis x={2},
                ymode=log,
                log basis y={2},
                xtick={1,2,4,8,16,32,64,128,256},
                xlabel={cores},
                ylabel={time (s)}
            ]
            \addplot table {./graphs/mergesort_min_time_disDCVecChunk_16777216_1000.data};
            \addlegendentry{Merge Sort disDCVecChunk $2^{24}$}
            \addplot table {./graphs/mergesort_min_time_disDCVecChunk_33554432_1000.data};
            \addlegendentry{Merge Sort disDCVecChunk $2^{25}$}
            \addplot table {./graphs/bitonic-skel_min_time_vector_16777216-.data};
            \addlegendentry{Bitonic Sort Vector $2^{24}$}
            \addplot table {./graphs/bitonic-skel_min_time_vector_33554432-.data};
            \addlegendentry{Bitonic Sort Vector $2^{25}$}
            \nextgroupplot[
                ymin=0,xmin=0,
                cycle list name=colorlist,
                legend style={align=left,anchor=west},
                ylabel style={align=left},
                width=0.5\textwidth,
                height=0.275\textwidth,
                xtick={32,64,128,256},
                xlabel={cores},
                ylabel={speedup}
            ]
            \addplot table {./graphs/mergesort_abs_speedup_min_time_disDCVecChunk_16777216_1000.data};
            \addplot table {./graphs/mergesort_abs_speedup_min_time_disDCVecChunk_33554432_1000.data};
            \addplot table {./graphs/bitonic-skel_abs_speedup_min_time_vector_16777216-.data};
            \addplot table {./graphs/bitonic-skel_abs_speedup_min_time_vector_33554432-.data};
        \end{groupplot}
    \end{tikzpicture}
    \caption{Bitonic Sorter / merge sort hybrid compared to a traditional merge sort on the Beowulf Cluster.}
    \label{fig:beowulf_time1}
\end{figure}
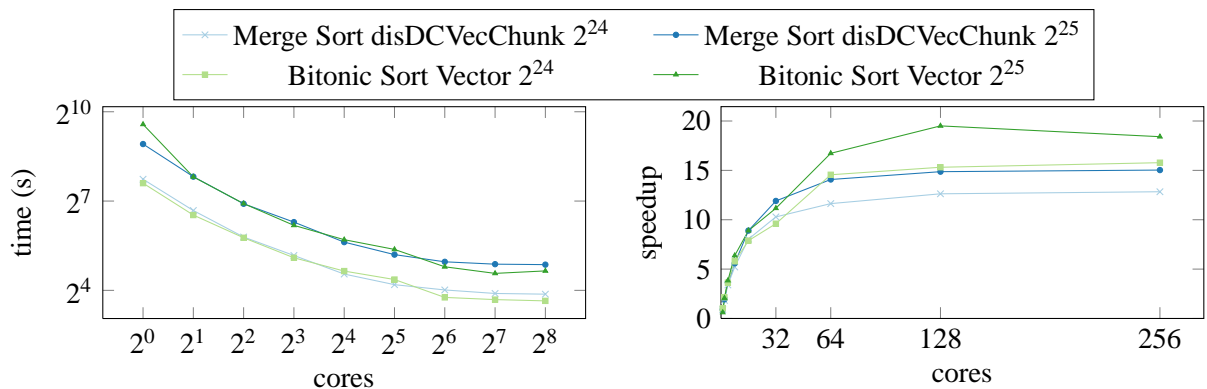

On the Beowulf Cluster the communication between different \cd{PEs} located on the same computer is cheap while intercommunication between computers is comparatively slow. 
The local communication structure of the bitonic sorting network is well-suited to this setting. The fixed communication structure of the Bitonic Sorter allows for an accurate process placement where the structure of the Bitonic Sorter is aligned to the structure of the cluster.

Another remarkable property of the bitonic sorting network is the potential of working with distributed input and output. The algorithm can work with distributed data without the need to aggregate the data. This is particularly interesting for very large sets of data. 
We will therefore compare the bitonic sorter to the PSRS algorithm~\cite{psrs93}, a parallel variant of quicksort with an elaborated pivot selection which guarantees a well-balanced distribution of the resulting lists. 
A comparison to PSRS is well-suited because the algorithmic structures are rather similar. 
In \autoref{fig:beowulf_time2} the runtime graphs of the PSRS algorithm and the Bitonic Sorter are depicted. The algorithms are modified to work with distributed data, only the sorting time without data distribution and collection is measured. 

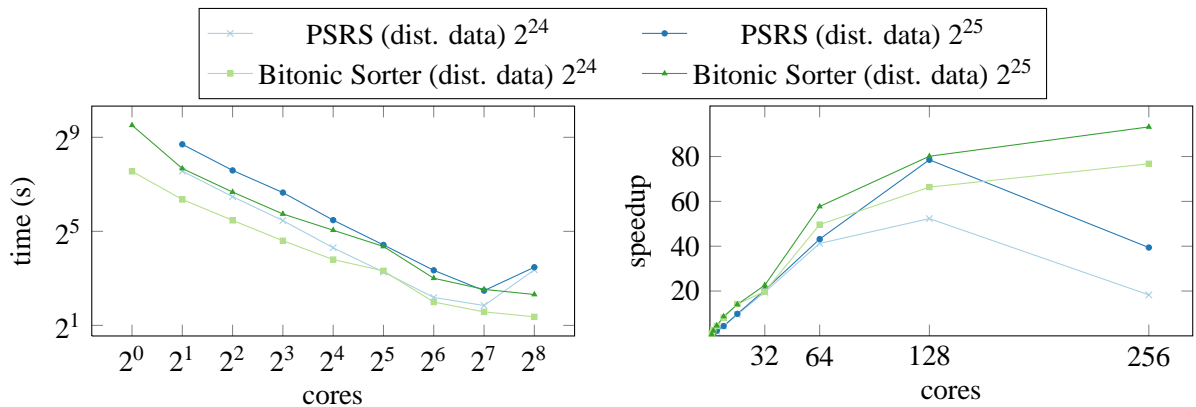
\begin{figure}[htbp]
    \centering
    \begin{tikzpicture}
        \begin{groupplot}[group style={group size=2 by 1, horizontal sep=1.8cm}]
            \nextgroupplot[
                cycle list name=colorlist,
                legend style={align=left,anchor=south,at={(1.1,1.03)},/tikz/every even column/.append style={column sep=0.5cm}},
                legend columns=2,
                ylabel style={align=left},
                width=0.5\textwidth,
                height=0.29\textwidth,
                xmode=log,
                log basis x={2},
                ymode=log,
                log basis y={2},
                xtick={1,2,4,8,16,32,64,128,256},
                xlabel={cores},
            ylabel={time (s)}]
            \addplot table {./graphs/mergesort_min_time_psrsDistribCTimeV_16777216_1000.data};
            \addlegendentry{PSRS (dist. data) $2^{24}$}
            \addplot table {./graphs/mergesort_min_time_psrsDistribCTimeV_33554432_1000.data};
            \addlegendentry{PSRS (dist. data) $2^{25}$}
            \addplot table {./graphs/bitonic-skel_min_time_bsDistribCTime_16777216-.data};
            \addlegendentry{ Bitonic Sorter (dist. data) $2^{24}$}
            \addplot table {./graphs/bitonic-skel_min_time_bsDistribCTime_33554432-.data};
            \addlegendentry{ Bitonic Sorter (dist. data) $2^{25}$}
            \nextgroupplot[
                ymin=0,xmin=0,
                cycle list name=colorlist,
                legend style={align=left,anchor=west},
                ylabel style={align=left},
                width=0.5\textwidth,
                height=0.29\textwidth,
                xtick={32,64,128,256},
                ytick={20,40,60,80},
                xlabel={cores},
                ylabel={speedup}
            ]
            \addplot table {./graphs/mergesort_abs_speedup_min_time_psrsDistribCTimeV_16777216_1000.data};
            \addplot table {./graphs/mergesort_abs_speedup_min_time_psrsDistribCTimeV_33554432_1000.data};
            \addplot table {./graphs/bitonic-skel_abs_speedup_min_time_bsDistribCTime_16777216-.data};
            \addplot table {./graphs/bitonic-skel_abs_speedup_min_time_bsDistribCTime_33554432-.data};
        \end{groupplot}
    \end{tikzpicture}
    \caption{Runtime and Speedup of merge sort with the Bitonic Sorter as merge stage compared to a traditional merge sort on the Beowulf Cluster.} 
    \label{fig:beowulf_time2}
\end{figure}

Again, the bitonic sorter scales well in comparison to the PSRS algorithm. With an increasing number of PEs the all-to-all communication of the PSRS algorithm becomes more expensive. % while it performs better with increasing data volume because of the singular transmission of data blocks. % (cf. \autoref{fig:beowulf_speedup2}).  

%\newpage % fixme fixme fixme
\section{Related Work}
\label{relwork}
There have been some newer approaches to sorting networks often in combination with hardware accelerators like FPGAs~\cite{mueller2012} or GPUs~\cite{gov2005}. In particular GPGPU programming has led to a little renaissance of sorting networks, especially with different implementations of the Bitonic Sorter~\cite{purcell2005,gov2006,kipfer2005} achieving good results. However these approaches usually either implement the bitonic sorter in the original way as presented by Batcher or sometimes implement the Adaptive Bitonic Sorter~\cite{bilardi1989} instead. The latter is a data dependent variant of the Bitonic Sorter and therefore not a sorting network. Consequently the work presented in this paper is closer to the different approaches of hybrid sorting algorithms. 
There are numerous examples for the benefit of hybrid sorting algorithms, for example in \cite{sintorn2008} a hybridization of Bucketsort and merge sort yields good results. 
Some ideas of this work were motivated by Dieterle's~\cite{mischa2016} work on skeleton composition.  

\section{Conclusion and Future Work}
\label{conc}
We have presented a different approach of agglomeration for sorting networks. This technique equips us with the possibility of using sorting networks as a parallel merging stage for arbitrary sorting algorithms, which is a versatile, easily adaptable and very promising approach. We are convinced that further improvements to the given example application are possible. 
We will investigate different possibilities of constructing combinations of arbitrary sorting algorithms with sorting networks. Therefore we will consider possible connections to embedded languages that allow for GPGPU programming from Haskell such as Accelerate~\cite{accelerator}
%\footnote{url: \cd{https://github.com/AccelerateHS/accelerate/}} 
or Obsidian~\cite{svensson08}
%\footnote{url: \cd{https://github.com/svenssonjoel/Obsidian}} 
or the possibility of combining the concise and easy to maintain functional implementation of sorting networks with efficient sorting algorithms for example via Haskell's Foreign Function Interface. 
Furthermore, most of the findings of this paper are applicable to other sorting networks such as Batcher's Odd-Even-Mergesort.  
All further investigations could benefit from a cost model that enables for better runtime predictions.  

\section*{Acknowledgments} 
The author thanks \mbox{Rita} \mbox{Loogen} and the anonymous reviewers for their helpful comments on a previous version of this paper and \mbox{Wolfgang} \mbox{Loidl}, \mbox{Prabhat} \mbox{Totoo} and \mbox{Phil} \mbox{Trinder} for giving us access to their Beowulf Cluster.

\bibliographystyle{eptcs}
\bibliography{Literaturverzeichnis}

\appendix
%\section*{Appendices} 
%
%\section{A simple, balanced {\ttfamily mergeSplit} for comparison elements}
\end{document}